\documentclass{article}

\usepackage{arxiv}

\usepackage[utf8]{inputenc} 
\usepackage[T1]{fontenc}    
\usepackage{hyperref}       
\usepackage{url}            
\usepackage{booktabs}       
\usepackage{amsfonts}       
\usepackage{nicefrac}       
\usepackage{microtype}      
\usepackage{lipsum}		
\usepackage{graphicx}
\usepackage{natbib}
\usepackage{doi}

\usepackage{times}
\usepackage{soul}
\usepackage{amsfonts}
\usepackage{url}
\usepackage[utf8]{inputenc}
\usepackage[small]{caption}
\usepackage{graphicx}
\usepackage{amsmath}
\usepackage{amsthm}
\usepackage{booktabs}
\usepackage{algorithm}
\usepackage{algorithmic}
\usepackage{color}
 \usepackage{booktabs}

\usepackage{float}
\usepackage{graphicx}
\usepackage{subcaption}

\graphicspath{ {./images/} }

\newtheorem{innercustomthm}{Theorem}
\newenvironment{customthm}[1]
{\renewcommand\theinnercustomthm{#1}\innercustomthm}
{\endinnercustomthm}

\newtheorem{theorem}{Theorem}
\newtheorem{lemma}{Lemma}

\title{Majority Vote in Social Networks: Make Random Friends or Be Stubborn to Overpower Elites}

\author{ \href{https://orcid.org/0000-0000-0000-0000}{\hspace{1mm}Charlotte Out} \\
	Department of Computer Science\\
	ETH Zürich\\
	
	\texttt{chaout@student.ethz.ch} \\
	\And
	\href{https://orcid.org/0000-0000-0000-0000}{\hspace{1mm}Ahad N. Zehmakan} \\
	Department of Computer Science\\
	ETH Zürich\\
	\texttt{ahadn.zehmakan@gmail.com} \\
}

\date{}

\renewcommand{\headeright}{}
\renewcommand{\undertitle}{}
\renewcommand{\shorttitle}{Accepted at IJCAI 2021}

\hypersetup{
pdftitle={A template for the arxiv style},
pdfsubject={q-bio.NC, q-bio.QM},
pdfauthor={David S.~Hippocampus, Elias D.~Striatum},
}

\begin{document}
\maketitle

\begin{abstract}
Consider a graph $G$, representing a social network. Assume that initially each node is colored either black or white, which corresponds to a positive or negative opinion regarding a consumer product or a technological innovation. In the majority model, in each round all nodes simultaneously update their color to the most frequent color among their connections.

Experiments on the graph data from the real world social networks (SNs) suggest that if all nodes in an extremely small set of high-degree nodes, often referred to as the elites, agree on a color, that color becomes the dominant color at the end of the process. We propose two countermeasures that can be adopted by individual nodes relatively easily
and guarantee that the elites will not have this disproportionate power to engineer the dominant output color. The first countermeasure essentially requires each node to make some new connections at random while the second one demands the nodes to be more reluctant towards changing their color (opinion). We verify their effectiveness and correctness both theoretically and experimentally.

We also investigate the majority model and a variant of it when the initial coloring is random on the real world SNs and several random graph models. In particular, our results on the Erd\H{o}s-R\'{e}nyi and regular random graphs confirm or support several theoretical findings or conjectures by the prior work regarding the threshold behavior of the process.

Finally, we provide theoretical and experimental evidence for the existence of a poly-logarithmic bound on the expected stabilization time of the majority model.
\end{abstract}


\section{Introduction}
When facing a decision or forming an opinion about a topic such as a consumer product,
a technological innovation, or a political event, humans often consult friends, family or others in their close circle for advice. Additionally, we often consider the opinions of the figures whose opinions we value in some way; for example, politicians we usually agree with, celebrities whom we look up to, or well-established scientists. In this way, an individual's opinion is influenced by the opinions of the people around her. Furthermore, due to the rise of online social networking, opinions are formed and changed at a higher pace. Consequently, there has been a growing demand for a quantitative understanding of the opinion forming process.

Recently, within the field of computer science, especially computational social choice and algorithmic game theory, there has been a rising interest in developing and studying mathematical opinion diffusion models, which aim to mimic the process of opinion forming in a society. 
At a high level of abstraction, in these models one usually consider a graph $G$ and some initial coloring of the nodes, where each node is colored either black or white. This graph is meant to represent a social network, in which the agents are modeled as nodes and an edge between two nodes corresponds to a relation between the respective agents, e.g. friendship, common interests, or advice. The color of a node represents its opinion on an innovation or a political party, etc. After initialization, in each round a group of nodes update their color based on a predefined rule.

Plentiful instances of the aforementioned abstract model have been introduced and studied. Among them, the majority model has attracted considerable attention, cf.~\citep{auletta2015minority} and ~\citep{gartner2020threshold}. In the \emph{majority model}, for a graph $G$ and an initial coloring, in each round all nodes simultaneously update their color to the most frequent color in their neighborhood. In case of a tie, a node keeps its current color. We also consider the \emph{($\psi_1$, $\psi_2$)-majority model}, for some $\psi_1, \psi_2>1/2$. Here, a black (resp. white) node changes its color if at least $\psi_1$ (resp. $\psi_2$) fraction of its neighbors hold the opposite color from itself. We observe that this is the same as the majority model for $\psi_1=\psi_2=\psi$ for a $\psi$ slightly larger than 1/2.

Several different variants of the majority model have been studied by prior work (cf.~\cite{keller2014even} and ~\cite{gartner2018majority}).
We will focus on the following two variants. Assume that each node $v$ has an \emph{influence factor} $r(v)$. Here also each node chooses the majority color, but it counts the color of a neighbor $v$, $r(v)$ times. By default we assume that $r(v)=1$ for each node $v$, otherwise it is mentioned explicitly. Note that if all nodes have influence factor one, we recover the majority model. Secondly, we consider the variant in which we assign a \emph{stubbornness factor} $\gamma(v)\in (0,1)$ to each node $v$. Then, a node $v$ changes its color if at least $\gamma(v)$ fraction of its neighbors have the opposite color. We observe that if we assign a fixed stubbornness factor $\gamma$ to all nodes, this would coincide with the $(\psi_1, \psi_2)$-majority model for $\psi_1=\psi_2=\gamma$. 

All the updating rules that we study in this paper are deterministic. Furthermore, for an $n$-node graph $G$, there are $2^n$ possible colorings. Therefore, after at most $2^n$ rounds the process reaches a cycle of colorings. The number of rounds the process needs to reach the cycle and the length of the cycle are called the \emph{stabilization time} and \emph{period} of the process. We say a color (black or white) \emph{wins} if more than half of the nodes share that color in the final configuration. Moreover, a color \emph{takes over} if all nodes share that color in the final configuration. We say a set of nodes form a white/black \emph{coalition} if they are all white/black. A node set $S$ is said to be a \emph{winning set} (resp. \emph{dynamo}) if black color wins (resp. takes over) once nodes in $S$ form a black coalition.

\subsection{Our contribution}
\paragraph{Question 1.}\textit{ What is the number of black nodes required for the black color to win or take over?}

Alternatively, more realistically speaking, if a marketing campaign can convince a group of individuals to adopt a new product, 
and the goal is to trigger a large cascade of further adoptions building on collective decision-making, which set of individuals should it target and how large this set needs to be? When considering graphs of real world Social Networks (SNs), ~\cite{MAIN_PELEG}
experimentally observed that in such graphs, if a small set of nodes (e.g. $1\%$ of nodes) with the highest degrees form a black coalition and have an influence factor slightly larger than the rest of nodes, the black color wins. Such a small set of high-degree nodes are meant to approximate the \emph{elites}, which are a relatively small and well-connected set of individuals (i.e., nodes) with substantial economic or social power, cf.~\cite{avin2017elites}. 

Several random graph models have been introduced to simulate the real world SNs. Arguably, one of the most well-studied such graph models is the Preferential Attachment (PA) random graph~\cite{barabasi_PA_graph}. In contrast to the real world SNs,~\cite{MAIN_PELEG} observed that in the PA graphs, with comparable number of nodes and edges, a small set of black nodes is not capable of enforcing the victory of the black color, unless they have extremely large influence factors. Therefore, they suggested for future work to propose graph models which not only retain well known characteristics of the real world SNs, but also support the existence of a small set of nodes with a significant disproportionate power in the majority model.

The fact is even though the PA model possesses some crucial features of the real world SNs, it also suffers from lack of some fundamental properties such as a high clustering coefficient (that is, two neighbors of a node are likely to be adjacent), cf.~\cite{krioukov2010hyperbolic}. Hence, the aforementioned result by~\cite{MAIN_PELEG} is just another indication that the PA model is not a decent choice to represent the SNs. On the other hand, Hyperbolic Random Graph (HRG) cf.~\cite{krioukov2010hyperbolic} is known to possess all the desired fundamental properties, including a high clustering coefficient, making it an interesting and suitable model to consider, leading to our first contribution.

\paragraph{Contribution 1.} Our experiments on HRG matches the results for the real world SNs; that is, almost the same number of highest degree black nodes with a rather small influence factor suffices for the black color to win. (The parameters of the HRG selected suitably to make it comparable to the respective SN. See Section~\ref{experimental setup}, for more details.)

Given this disproportionate amount of controlling power of the elites, a natural question that arises is whether one can develop a counter measure to overpower them. A proposed countermeasure ideally should not require significant changes in the graph structure or the updating rule. Furthermore, it should be easy for the agents to implement, e.g. it should not require them to have a full knowledge of the graph structure or memorize the history of the process.

\paragraph{Contribution 2.} We propose two countermeasures and support their effectiveness both experimentally and theoretically. Firstly, we show that if we require every agent to make a certain number of new connections at random, a small set of elite nodes cannot control the output of the majority model anymore. More formally, we prove that if we add a graph with strong expansion properties on top of any graph, including a real world SN, for the black color to win a significant number of nodes must be black initially. Secondly, we show that if each agent changes her color only if a ``sufficiently" large fraction of her neighbors hold a different color, no small set of nodes can determine the output of the process. This is realized by assigning a high stubbornness factor to each node. Furthermore, we demonstrate that our experiments support these theoretical findings. 

\paragraph{Question 2.} \textit{What is the probability that the black color wins (or takes over) if we assume that initially each node is colored black independently with some probability $p_b\in (0,1)$?}

It is known (cf.~\cite{ZEHMAKAN_opinionformingEREXPAND}) that in the majority model on ``nearly" regular graphs with strong expansion properties, such as random regular graphs and Erd\H{o}s-R\'{e}nyi random graph (See Section~\ref{GraphDef} for a formal definition), if $p_b$ is ``slightly" more (resp. less) than 1/2, then black (resp. white) color takes over asymptotically almost surely. (We say that an event occurs asymptotically almost surely (a.a.s.) if it happens with probability tending to 1 when we let the number of nodes go to infinity.)

\paragraph{Contribution 3.} We provide experimental results confirming the existence of this threshold behavior in the Erd\H{o}s-R\'{e}nyi and regular random graphs, and show that this behavior is also observed in the PA random graph. However, the majority model turns out to exhibit a different behavior on real world SNs and HRG. That is, the black (resp. white) color might not \emph{take over} even when $p_b$ (resp. $1 - p_b$) is significantly larger than 1/2. However, we show that upon the addition of a $d$-regular random graph for a reasonably large $d$ we recover the aforementioned threshold behavior. 

We study the above question for the $(\psi_1, \psi_2)$-majority model as well. We prove that for a dense Erd\H{o}s-R\'{e}nyi random graph, the process exhibits a threshold behavior with two phase transitions: (i) the white color \emph{takes over} if $p_b < 1- \psi_1$ (ii) both colors will \emph{survive} (i.e., no color takes over) if $1-\psi_1 < p_b < \psi_2$ (iii) the black color \emph{takes over} if $\psi_2< p_b$ a.a.s. Furthermore, our experiments suggest that such a threshold behavior is also present in the PA random graph, sparse regular random graphs, HRG, and the real world SNs, but for different threshold values.

As mentioned, the majority model on the Erd\H{o}s-R\'{e}nyi random graph $\mathcal{G}_{n,q}$ is well understood when $p_b$ is smaller or larger than $1/2$. What if we have $p_b=1/2$?~\cite{CONJECTURE_convergence_benjamini} conjectured if $q$ is ``sufficiently'' larger than $1/n$, then a.a.s. one of the two colors \emph{almost takes over} (i.e., all nodes share the same color at the end, except a sub-linear number of them).~\cite{fountoulakis2020resolution} proved that the conjecture is true when $q$ is larger than $1/\sqrt{n}$, but it has remained open for $q$ smaller than $1/\sqrt{n}$.

\paragraph{Contribution 4.} We perform experiments whose results support this conjecture. More precisely, we observe that in the majority model on $\mathcal{G}_{n,q}$ with $p_b=1/2$ if $q=c/n$ for $c \geq 12$, one of the two colors almost takes over. However, if $c \leq 8$, the process reaches a configuration where almost half of the nodes are black.

\paragraph{Question 3.}\textit{ What is the stabilization time and period of the process?}

For the majority model on a graph $G=(V,E)$,~\cite{GOLES1980187} proved that the period is always one or two. Furthermore using some algebraic tools,~\cite{poljak1986pre} showed that the stabilization time is in $\mathcal{O}(|E|)$, which ~\cite{KELLER_frischknecht2013convergence} proved to be tight, up to some poly-logarithmic factor. However, if we start from a random coloring, where each node is black independently with probability $p_b$, the probability that an extremal coloring, for which the process takes a long time, emerges is fairly small. Hence, a natural question that arises here is whether one can provide stronger bounds on the \emph{expected stabilization time} (i.e., the expected number of rounds the process needs to reach a cycle of colorings from a random initial coloring). It is widely believed that a poly-logarithmic upper bound must exist, but this is proven only for some special classes of graphs, cf.~\cite{ZEHMAKAN_opinionformingEREXPAND}. 

\paragraph{Contribution 5.} As our first evidence for a poly-logarithmic upper bound, we prove that the expected stabilization time of the majority model is at most $\log n$ when the underlying graph is a cycle $C_n$. Furthermore, we experimentally investigate the expected stabilization time of the majority model on different random graph models and real world SNs and our findings support the conjectured poly-logarithmic bound. It is worth to stress that the process takes the longest at the threshold value $p_b=1/2$.

We study the stabilization time of the ($\psi_1$, $\psi_2$)-majority model too and, building on a potential function argument, prove that it is also bounded by $\mathcal{O}(|E|)$ when $\psi_1=\psi_2$. For the proof, we set a connection between the number of edges whose endpoints have opposite colors in the initial coloring and the number of rounds the process needs to end. As we will explain, this technique might be useful to prove a poly-logarithmic bound on the expected stabilization of the majority and ($\psi_1$, $\psi_2$)-majority model.

\subsection{Preliminaries}
\label{Preliminaries}

\subsubsection{Graph Definitions.}
\label{GraphDef}
Let $G=\left(V,E\right)$ be an $n$-node graph.  
For a node $v\in V$, $N\left(v\right):=\{u\in V: \{u,v\} \in E\}$ is the \emph{neighborhood} of $v$. For a set $S\subset V$, we define $N\left(S\right):=\bigcup_{v\in S}N\left(v\right)$ and $N_S\left(v\right):=N\left(v\right)\cap S$. Moreover, $d\left(v\right):=|N\left(v\right)|$ is the \emph{degree} of $v$ in $G$ and $d_S\left(v\right):=|N_S\left(v\right)|$. Furthermore, for two node sets $S$ and $S'$, we define $e\left(S,S'\right):=|\{\left(v,u\right)\in S\times S': \{v,u\}\in E\}|$
where $S\times S'$ is the Cartesian product of $S$ and $S'$. Note that whenever graph $G$ is not clear from the context, we add a superscript, e.g. we write $d^G(v)$, $d^G_S(v)$, and $e^G(S,S')$.

\paragraph{Random Graphs.} Let $\mathcal{G}_{n,q}$ denote the Erd\H{o}s-R\'{e}nyi random graph, which is the
random graph on the set $\{1,\cdots, n\}$, where each edge is present independently with probability $q$. We denote by $\mathcal{G}_{n,d}$ the $d$-regular random graph, which is the random graph with a uniform distribution over all $d$-regular graphs on $n$ nodes.

\subsubsection{Models.} For a graph $G=\left(V,E\right)$, a \emph{coloring} is a function $\mathcal{C}:V\rightarrow\{b,w\}$, where $b$ and $w$ represent black and white, respectively. For a node $v\in V$, the set $N_a^{\mathcal{C}}\left(v\right):=\{u\in N\left(v\right):\mathcal{C}\left(u\right)=a\}$ includes the neighbors of $v$ which have color $a\in\{b,w\}$ in coloring $\mathcal{C}$. 
Assume that we are given an initial coloring $\mathcal{C}_0$ on a graph $G$. In a model $M$,  $\mathcal{C}_t\left(v\right)$, which is the color of node $v$ in the $t$-th coloring for $t\in\mathbb{N}$, is determined based on a predefined updating rule. We are mainly interested in the following two models, where $\mathcal{C}_t\left(v\right)$ is defined by a deterministic updating rule as a function of $\mathcal{C}_{t-1}\left(u\right)$ for $u\in N\left(v\right)\cup \{v\}$.

\paragraph{Majority Model.} In the majority model

$\mathcal{C}_t(v)$ = 
$\begin{cases} 
\mathcal{C}_{t-1}(v) \quad if
|N_b^{\mathcal{C}_t-1}(v)| = |N_w^{\mathcal{C}_{t-1}}(v)|\\
argmax_{a \in \{b, w\}}|N_a^{\mathcal{C}_{t-1}}(v)| \quad otherwise\\
\end{cases}.$

\paragraph{($\psi_1$, $\psi_2$)-Majority Model.}
In the $(\psi_1, \psi_2)$-majority model for some $\psi_1,\psi_2\in (1/2,1]$

$\mathcal{C}_t(v)$ =
$\begin{cases}
w & if \mathcal{C}_{t-1}(v)=b \wedge |N_w^{\mathcal{C}_{t-1}}(v)| \ge \psi_1d(v)\\
b & if \mathcal{C}_{t-1}(v)=w \wedge |N_b^{\mathcal{C}_{t-1}}(v)| \ge \psi_2d(v)\\
\mathcal{C}_{t-1}(v) & otherwise
\end{cases}.$

In these models, we define $B_t$ and $W_t$ for $t\in \mathbb{N}_0$ to be the set of black and white nodes in $\mathcal{C}_t$. 

\subsubsection{Experimental Setup.}
\label{experimental setup} We run our experiments for the graph data of the Facebook (FB) and YouTube (YT) SN from~\cite{viswanath-2009-activity_FBMISLOVE}
and \cite{mislove-2007-socialnetworks} and Twitter (TW) and Slashdot (SD) graph data from~\cite{snapnets}, of which we give a short description, as well as some basic graph properties of these data sets summarized Table 1. We should emphasize that for all these datasets, in the case they contained directed edges, we converted them into undirected graphs by simply treating directed edges as undirected.

\textbf{YouTube (YT):} A video sharing website on which people can follow (subscribe to) other users. An edge in the social network graph represents a following relationship.\\
\textbf{Slashdot (SD):} A social news website with primarily technology-related news. In 2002, Slashdot introduced the Slashdot   
Zoo feature which allows users to tag each other as friends or foes. The network contains friend/foe links between the users of Slashdot. \\
\textbf{Twitter (TW):} A microblogging and social networking service. Here users post and interact with messages called ``tweets". A directed edge in the graph represents a following relationship, i.e. an edge from node $v$ to node $u$ implies that user $v$ follows user $u$. \\
\textbf{Facebook (FB):} A social media and social networking service. The graph consists of all user-to-user links from the Facebook New Orleans networks. An edge in the graph between node $v$ and node $u$ hence indicates that user $v$ and user $u$ are friends.

\begin{table}[!ht]
\label{table social nets}
\centering
\begin{tabular}{lrrr} 
\midrule 
SN name & n & m & ave. deg.\\ 
\midrule 
YT & 1138499 & 2990443 & 5.25 \\
SD & 82168 & 582533 & 14.18 \\
TW & 81306 & 1342310 & 33.02\\ 
FB & 63731 & 817090 & 25.64\\
\bottomrule
\end{tabular}
\vspace{2mm}
\caption{Some basic graph properties of the four examined SNs, where $n$ and $m$ denote the number of nodes and edges, respectively.}
\end{table}

Furthermore, we focus on several random graph models such as Erd\H{o}s-R\'{e}nyi (ER) Random Graph, Random Regular Graph (RRG), Preferential Attachment (PA) Random Graph, and Hyperbolic Random Graph (HRG). To make our experiments on the random graph models and real world SNs comparable, we set the parameters of random graphs in a way that they have the same number of nodes and edges in expectation. For the generation of the RRG and HRG, we rely on the (approximation) algorithms of ~\cite{steger1999generating} and ~\cite{staudt2015networkit} respectively, and the implementations in \cite{networkx}. To generate HRG, in addition to the number of nodes and edges, one needs to provide the exponent of the power-law degree distribution $\beta$ and the temperature $T$ as the input parameters. Throughout this paper, we set $\beta=2.5$ and $T=0.6$. Experiments which required random choice of edges or colors were executed 8 times and then the average output was considered. Furthermore, (several) experiments were carried out on an Intel Xeon E3 CPU, with 32 GB RAM, and a Linux OS.

\subsection{Prior Works}

\paragraph{Opinion Diffusion Models.} 
In the plethora of opinion diffusion models, 
considerable attention has been devoted to the study of different variants of the majority model, such as an asynchronous updating rule~\cite{DevilDetails} and~\cite{zehmakan2019tight}, various tie-breaking rules~\cite{schoenebeck2018consensus} and ~\cite{jeger2019dynamic}, and randomized updating rules~\cite{Mossel_2013} and~\cite{n2020rumor}. Even more complex models such as the one considered in~\cite{ferraioli2017social} 
which follows an averaging-based updating rule, or the models in~\cite{faliszewski2020opinion} and~\cite{brill2016pairwise} can be seen as extensions of the majority model. In the present paper, we consider variants of the majority model previously studied by cf. ~\cite{MAIN_PELEG} and ~\cite{Stubborn_nodes_Auletta}. Furthermore, the ($\psi_1$, $\psi_2$)-majority model, which is a generalization of the model studied in~\cite{balister2010random} and ~\cite{balogh2009majority}.

\paragraph{Minimum Size of a Winning Set.} Determining the minimum size of a winning set and a dynamo has been considered in various majority based models and for different classes of graphs such as ER~\cite{CONJECTURE_convergence_benjamini}, PA~\cite{MAIN_PELEG}, and lattice~\cite{balister2010random},~\cite{balogh2012sharp}, and~\cite{gartner2017color}. For general graphs,~\cite{Berger_CONSTANTSIZE_MONOP} proved that there exist arbitrarily large graphs which have dynamos
of constant size under the majority model and it was shown in~\cite{auletta2018reasoning} that every $n$-node graph has a dynamo of size at most $n/2$ under the asynchronous variant.

\paragraph{Random Initial Coloring.} The majority model with a random initial coloring has been investigated for different classes of graphs such as hypercubes and preferential attachment trees, cf.~\cite{balister2010random} and ~\cite{zehmakan2019two}. As stated, special attention has been devoted to the study of ER when each node is black independently with probability $p_b=1/2$, cf.~\cite{CONJECTURE_convergence_benjamini} and~\cite{shimizu2020quasi}. 

\paragraph{Stabilization Time and Period.} \cite{GOLES1980187} proved that the period of the majority model is always one or two. Recently, it was shown by~\cite{chistikov2020convergence} that it is PSPACE-complete to decide whether the period is one or not for a given coloring of a \emph{directed} graph. Furthermore, \cite{poljak1986pre} proved that the stabilization time of the majority model on a graph $G=(V,E)$ is upper-bounded by $\mathcal{O}(|E|)$. Stronger bounds are known for special classes of graphs. For instance, for a $d$-regular graph with strong expansion properties the stabilization time is in $\mathcal{O}(\log_d n)$, cf.~\cite{ZEHMAKAN_opinionformingEREXPAND}. 

\section{Power of Elites and Countermeasures}
\label{Power-of-elites}

\cite{MAIN_PELEG} observed that in real world SNs, if a small set (e.g. $1\%$) of the elite nodes are provided with a constant influence factor (e.g. 8), they are capable of determining the outcome of the majority model, i.e., they form a winning set. In the PA random graphs with comparable parameters, in contrast, these authors showed that for a small set of elite nodes to form a winning set, they must have an extremely large influence factor. As mentioned earlier, the PA random graph lacks the presence of a high clustering coefficient cf.~\cite{krioukov2010hyperbolic}. We believe this is the source of such discrepancy.
On the other hand, HRG is known~\cite{krioukov2010hyperbolic} to possess all the aforementioned properties, justifying our choice to investigate the majority model on HRG.

Our experiments demonstrate that the size and influence factor required for a set of elites to form a winning set is approximately the same in the real world SNs and HRGs with comparable parameters. This is depicted for YT SN in Figure~\ref{fig:elite12} (left), and other SNs are included in the full version of this paper. Figure~\ref{fig:elite12} (left) also covers PA, which was already investigated by~\cite{MAIN_PELEG}. 

Naturally, the question arises how to prevent a small set of elite nodes from determining the outcome of the majority model. To this end, we propose two countermeasures to overpower the elites. Firstly, we prove that if we add a sufficiently dense RRG on top of any graph, in particular a SN, no small winning set will exist. Secondly, we show that if we assign a sufficiently large stubbornness factor to each node, no small set of elites can create a winning set. We support both countermeasures theoretically as well as experimentally. 

\paragraph{Countermeasure 1.} Adding a RRG on top of a SN is essentially the same as asking agents (nodes) to make a set of connections at random. Consider a small set $S$ of elite nodes who form a winning set in a SN. Intuitively speaking, the randomly added connections for each node are unlikely to be chosen from set $S$, thus reducing the influencing power of the elite nodes in $S$. Before stating our result formally in Theorem~\ref{counter measure 1}, let us present Lemma~\ref{expander mixing lemma}, which will be the main workhorse of the proof. For a graph $G$, let $\sigma(G)$ be the second-largest absolute eigenvalue of its normalized adjacency matrix, which is an algebraic measure of expansion. Loosely speaking, Lemma~\ref{expander mixing lemma} asserts that the number of edges between any two node sets is almost completely determined by their cardinality, given that $\sigma(G)$ is small (i.e., $G$ has strong expansion properties).

\begin{lemma}[\cite{friedman2003proof_ALONEIGVALUE}]
\label{expander mixing lemma}
For any two node sets $S,S'$ in a $d$-regular graph, $\left|e(S,S') - \frac{|S||S'|d}{n}\right| \leq \sigma d \sqrt{|S||S'|}$.
\end{lemma}
  
\setcounter{theorem}{1}  
\begin{theorem}
\label{counter measure 1}
Let $G_1=(V,E_1)$ be an arbitrary graph with average degree $\bar{d}$, $G_2=(V,E_2)$ be a $d$-regular graph and $Z \subset V$ be an arbitrary set of nodes in $G = (V, E_1 \cup E_2)$ with $|Z|= 0.05n$ and $n=|V|$.
Consider the majority model on $G$, where $B_0=Z$ and all nodes in $Z$ have influence factor $r\leq 10$ (while it is 1 for the rest of nodes). If $\sigma(G_2) \leq \beta$ and $d=cr\bar{d}$ for a suitable choice of constants $c,\beta>0$, the white color wins.
\end{theorem}
  
\begin{proof}
Let us define $L:= \{v \in V\setminus Z : d^{G_1}(v) \leq 10\bar{d}\}$ and $H:= V\setminus \{ L \cup Z\}$. We observe that $\left|H\right| \leq 0.1n$ since $10 \bar{d}\left|H\right| \leq \sum_{v\in H}d^{G_1}(v) \leq \sum_{v\in V}d^{G_1}(v) = n\bar{d}$. Note that $\left|L\right| \geq 0.85n$. We claim that if $|L\cap W_t|\ge 0.75 n$ for some $t\in \mathbb{N}_0$, then $|L\cap W_{t+1}|\ge 0.75 n$. This immediately implies that if $B_0=Z$, then the white color wins. (Note that since $|Z|=0.05 n$, we have $|L\cap W_0|=|L|\ge 0.85 n\ge 0.75 n$.)

Let $L_w:=L\cap W_t$, with $|L_w|\ge 0.75 n$, and $B:=L\cap B_{t+1}$.
To prove our claim, it suffices to show that $\left|B\right| \leq 0.1n$. For a node $v$ in $B$, it is necessary that $d_{L_w}^{G_2} (v)\leq rd_Z^{G_1}(v) + d_{Y}^{G_1}(v) + rd_Z^{G_2}(v) + d_Y^{G_2}(v)$, where $Y := V\setminus \{Z \cup L_w\}$. Since $v \notin H$, $rd_Z^{G_1}(v) + d_Y^{G_1}(v) \leq 10r\bar{d}$. The combination of these two inequalities and a summation over all nodes in $B$ gives us $e^{G_2}(B, L_{w}) \leq re^{G_2}(B,Z) + e^{G_2}(B,Y) + 10r\bar{d}\left|B\right|$. Applying Lemma~\ref{expander mixing lemma} to both sides of the inequality yields
\begin{align*}
  \frac{\left|B\right|\left|L_w\right|d}{n} - \sigma d\sqrt{\left|B\right|\left|L_w\right|} \leq \frac{r\left|B\right|\left|Z\right|d}{n} + r\sigma d\sqrt{\left|B\right|\left|Z\right|}
  + \frac{\left|B\right|\left|Y\right|d}{n} + \sigma d \sqrt{\left|B\right|\left|Y\right|} + 10r\bar{d}\left|B\right.|
\end{align*}
Dividing both sides by $\frac{\sqrt{\left|B\right|}d}{n}$, setting $d = cr\bar{d}$ and rearranging the terms give us $\sqrt{\left|B\right|}\left(\left|L_w\right| - r\left|Z\right| - \left|Y\right| - \frac{10n}{c}\right) \leq
  \sigma n (\sqrt{\left|L_w\right|} + r\sqrt{\left|Z\right|} + \sqrt{\left|Y\right|})$. Recall that $0.75n \leq \left|L_w\right| \leq n$, $\left|Z\right| =0.05n$, $\left|Y\right| \leq 0.2n$ (since $Y=V\setminus \{Z\cup L_w\}$), $\sigma \leq \beta$, and $r \leq 10$. By using these bounds and rearranging, we get

\begin{align*}
  \left|B\right| \leq \frac{\beta^2\left(\sqrt{n} + r\sqrt{0.05n} + \sqrt{0.2n}\right)^2}{(0.75 - 0.05r - 0.2 - 10/c)^2} \leq \frac{12^2\beta^2n}{(0.05 - 10/c)^2}
\end{align*}
Hence, $\left|B\right| \leq 0.1n$ for a sufficiently large constant $c$ and a sufficiently small constant $\beta >0$. It is worth to stress that we did not attempt to minimize constant $c$, as the proof is merely for theoretical purposes. As we will see in our experiments on real world SNs a much smaller choice of $c$, namely $c=2$, suffices. 
\end{proof}

\paragraph{Corollary.}~\cite{friedman2003proof_ALONEIGVALUE} proved that for a random $d$-regular graph $\mathcal{G}_{n,d}$, $\sigma\left(\mathcal{G}_{n,d}\right) \leq 2/ \sqrt{d}$ a.a.s. when $d \geq 3$. This implies that the statement of Theorem~\ref{counter measure 1} holds a.a.s if $G_2=\mathcal{G}_{n,d}$ for a sufficiently large $d$. Therefore, if we add a $\mathcal{G}_{n,d}$ on top of a SN, there is no winning set which includes less than $5\%$ of the nodes. Recall that based on ~\cite{MAIN_PELEG} the real world SNs usually allow winning sets of much smaller size than $5\%$; for example in YT SN, a set of highest degree nodes of size $0.15\%$. As we will discuss, our experiments support even stronger bounds than the one given in Theorem $\ref{counter measure 1}$.
Lastly, it is worth to mention that the constraint $r\le 10$ can be relaxed, but cannot be lifted entirely because if a set $S$ of elites in a SN have extremely large influence factors, even after adding a complete graph on top of the SN, $S$ is a winning set. 

\paragraph{Countermeasure 2.} 
Note that we consider the setting in which initially a set of elites form a black coalition, and the rest of nodes (or most of them) are white. Hence intuitively speaking, if most nodes become very reluctant to change their color (i.e., have a large stubbornness factor), one would expect most of the white nodes to keep their color unchanged. We state this observation more formally in Theorem~\ref{Counter measure 2}.

\begin{theorem}
\label{Counter measure 2}
Consider a graph $G=(V,E)$. Let $Z\subset V$ such that $|Z|< n/2$ and $d_Z(v)\le fd(v)$ for some $f\in (0,1)$ and every $v\in V\setminus Z$. Consider the majority model where all nodes in $Z$ have influence factor $r\in\mathbb{N}$. If for each node $v\in V\setminus Z$ the stubbornness factor $\gamma(v)> \frac{r}{r + \frac{1-f}{f}}$, then $Z$ is not a winning set. 
\end{theorem}

\begin{proof}
Assume that initially the nodes in $Z$ form a black coalition and the rest of nodes are white, that is, $B_0=Z$ and $W_0=V\setminus Z$. Note that if all nodes in $V\setminus Z$ are white in some round during the process, and for every node $v\in V\setminus Z$, it holds that
\begin{equation}
\label{equation-stubborn}
\left(1- \gamma\left(v\right)\right)rd_Z(v)< \gamma\left(v\right)d_{V\setminus Z}\left(v\right)
\end{equation}
then, all nodes in $V\setminus Z$ will stay white in the next round. By an inductive argument, we can conclude that the node set $V\setminus Z$ remains white forever. This implies that the black color does not win (i.e., $Z$ is not a winning set).

It remains to prove that for each node $v\in V\setminus Z$, Equation~(\ref{equation-stubborn}) holds. Since $d_Z(v)\le fd(v)$ and $d_{V\setminus Z}(v)\ge (1-f)d(v)$, it suffices to show that
\[
\left(1- \gamma\left(v\right)\right)rfd\left(v\right) <  \gamma\left(v\right)\left(1-f\right)d\left(v\right)
\]
which is true because $\gamma(v)> \frac{r}{r + \frac{\left(1-f\right)}{f}}$.
\end{proof}

Observe that the statement of Theorem~\ref{Counter measure 2}, in particular, is true when $Z$ includes the nodes of highest degree and they all have influence factor $r$ (while it is 1 for the rest of nodes).

\paragraph{Experiments of Countermeasures.} The experimental results of both countermeasures applied to FB SN are depicted in Figure~\ref{fig:elite12}~(right), which confirm their effectiveness. The plots for other SNs are given in Appendix \ref{Sec: Experiment 1}. For instance, our experiments indicate that in FB SN an elite set consisting of $0.4\%$ of nodes with influence factor $16$ form a winning set, but after applying countermeasure $CM_1$ (adding a RRG) and $CM_2$ (assigning high stubbornness factor) an elite set of size $10\%$ and $33\%$, respectively, is required to win, with the same influence factor. 
 
\begin{figure}[h]
\centering
\includegraphics[scale=0.27]{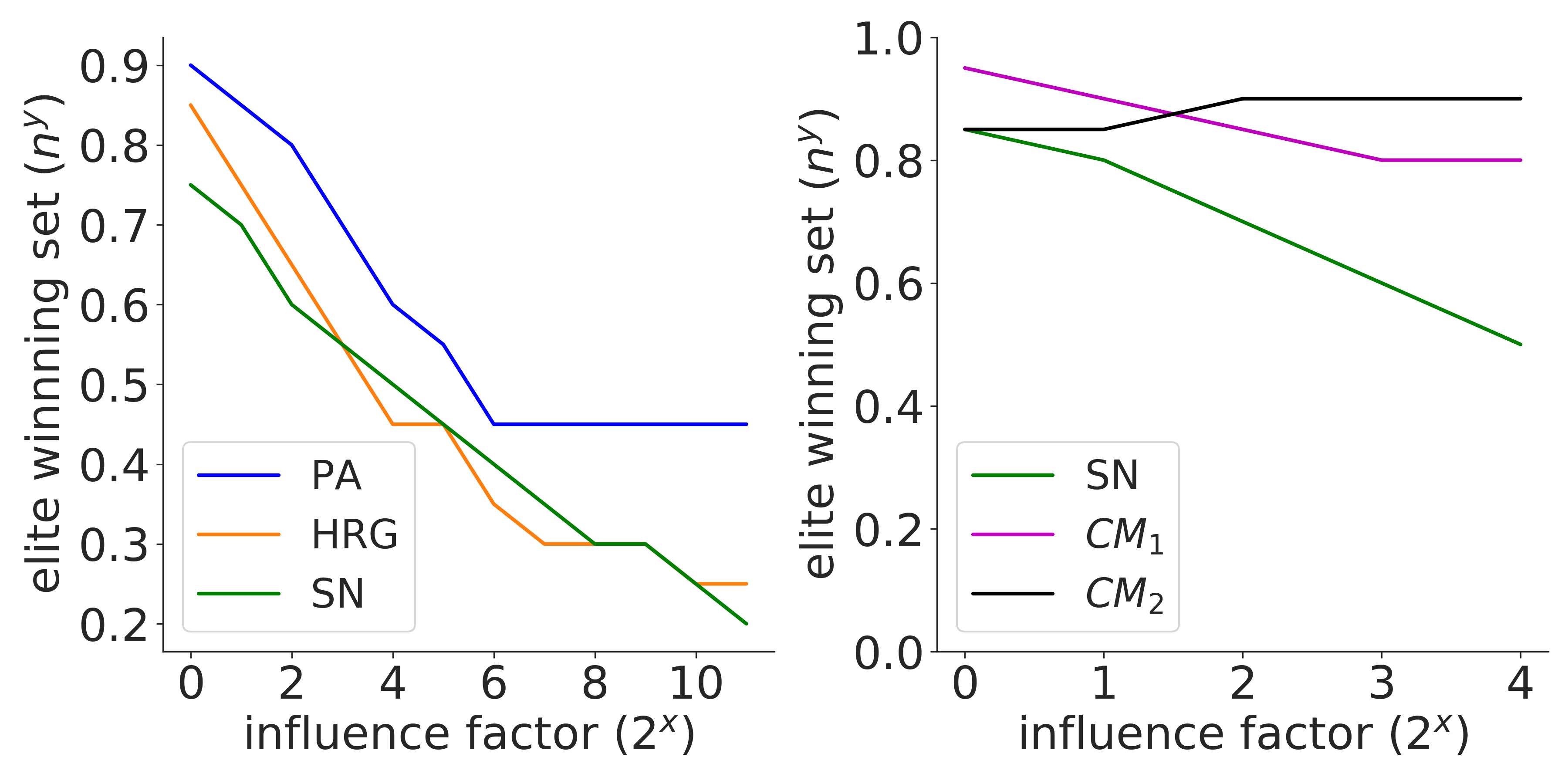}

\caption{The minimum size of a winning set of elite nodes whose influence factor is $r=2^x$ (while it is 1 for the rest of nodes) \textbf{(left)} on YT SN and HRG and PA with comparable parameters, and \textbf{(right)} on FB SN and the graphs $CM_1$ and $CM_2$, corresponding to our countermeasures applied to FB SN. $CM_1$ is the union of FB SN and a RRG with degree $d=2r\bar{d}$ (where $\bar{d}$ is the average degree of FB SN). $CM_{2}$ denotes FB SN in which all nodes get assigned a stubbornness factor $ \gamma = 1 - 1/2r$.}

\label{fig:elite12}
\end{figure}

\section{Random Initial Coloring}

We experimentally investigate the majority and ($\psi_1$, $\psi_2$)-majority model on various real world SNs and random graph models, where initially each node is black independently with probability $p_b$. We aim to determine the \emph{final fraction} of black nodes (i.e., the number of black nodes in the final configuration divided by the number of all nodes) for different values of $p_b$.

\paragraph{Majority Model with Random Coloring.} Our experimental results for the majority model on SD SN and several random graph models with comparable parameters are depicted in Figure~\ref{fig: density12}~(left). The plots for other real world SNs can be found Appendix \ref{Sec: Experiment 2}. We observe that the majority model is a \emph{fair density classifier} on ER, RRG, and PA. That is, for $p_b < 1/2$ (resp. $p_b > 1/2)$, the white (resp. black) color \emph{takes over}. This confirms some prior theoretical results, cf.~\cite{ZEHMAKAN_opinionformingEREXPAND}. However for SD SN and HRG, the white  (resp. black) color might not \emph{take over} even when $p_b$ (resp. $1-p_b$) is much smaller than $1/2$. Lastly, we note that upon the addition of a RRG with the same average degree on top of SD SN, the aforementioned fair density classification behavior emerges. (See $CM_1$ in Figure~\ref{fig: density12}~(left))

\paragraph{Uniform Random Coloring.} In prior work, special attention has been devoted to the study of the majority model on the Erd\H{o}s-R\'{e}nyi random graph $\mathcal{G}_{n,q}$ for $p_b=1/2$. In particular,~\cite{CONJECTURE_convergence_benjamini} conjectured that if $q$ is sufficiently larger (resp. smaller) than $1/n$, the process reaches an almost monochromatic coloring (resp. an almost balanced coloring) a.a.s. In an almost monochromatic coloring, all nodes share the same color, except a sub-linear number of them, and in an almost balanced coloring the difference between the number of black and white nodes is sub-linear. Our experiments on $\mathcal{G}_{n,q}$ for $n=1000000$ indicate that for $q\ge 12/n$ (resp. $q\le 8/n$) the process reaches an almost monochromatic coloring (resp. an almost balanced coloring). Hence, our results support the correctness of the conjecture.

\paragraph{($\psi_1, \psi_2$)-Majority Model with Random Coloring.} We prove that in the ($\psi_1, \psi_2$)-majority model on $\mathcal{G}_{n,q}$ a.a.s., for $q$ sufficiently larger than $\log n/n$ (which is the connectivity threshold): (i) the white color \emph{takes over} if $p_b < 1- \psi_1$ (ii) both colors will \emph{survive} (i.e., no color takes over) if $1-\psi_1 < p_b < \psi_2$ (iii) the black color \emph{takes over} if $\psi_2< p_b$. The proof of this proposition can be found in Appendix \ref{Sec: Two Phase Transition}. Furthermore, we experimentally investigate the $(\psi_1, \psi_2$)-majority model, for certain values of $\psi_1, \psi_2$, with a random initial coloring on various real world SNs and random graph models with comparable parameters. Our results for TW SN and corresponding random graphs are depicted in Figure~\ref{fig: density12}~(right). Similar plots are provided in Appendix \ref{Sec: Experiment 3}. We observe that a similar threshold behavior with two phase transitions is also present in TW SN, PA, HRG, and RRG but the threshold values are different from $1-\psi_1$ and $\psi_2$. 

\begin{figure}[h]
\centering
\includegraphics[scale=0.27]{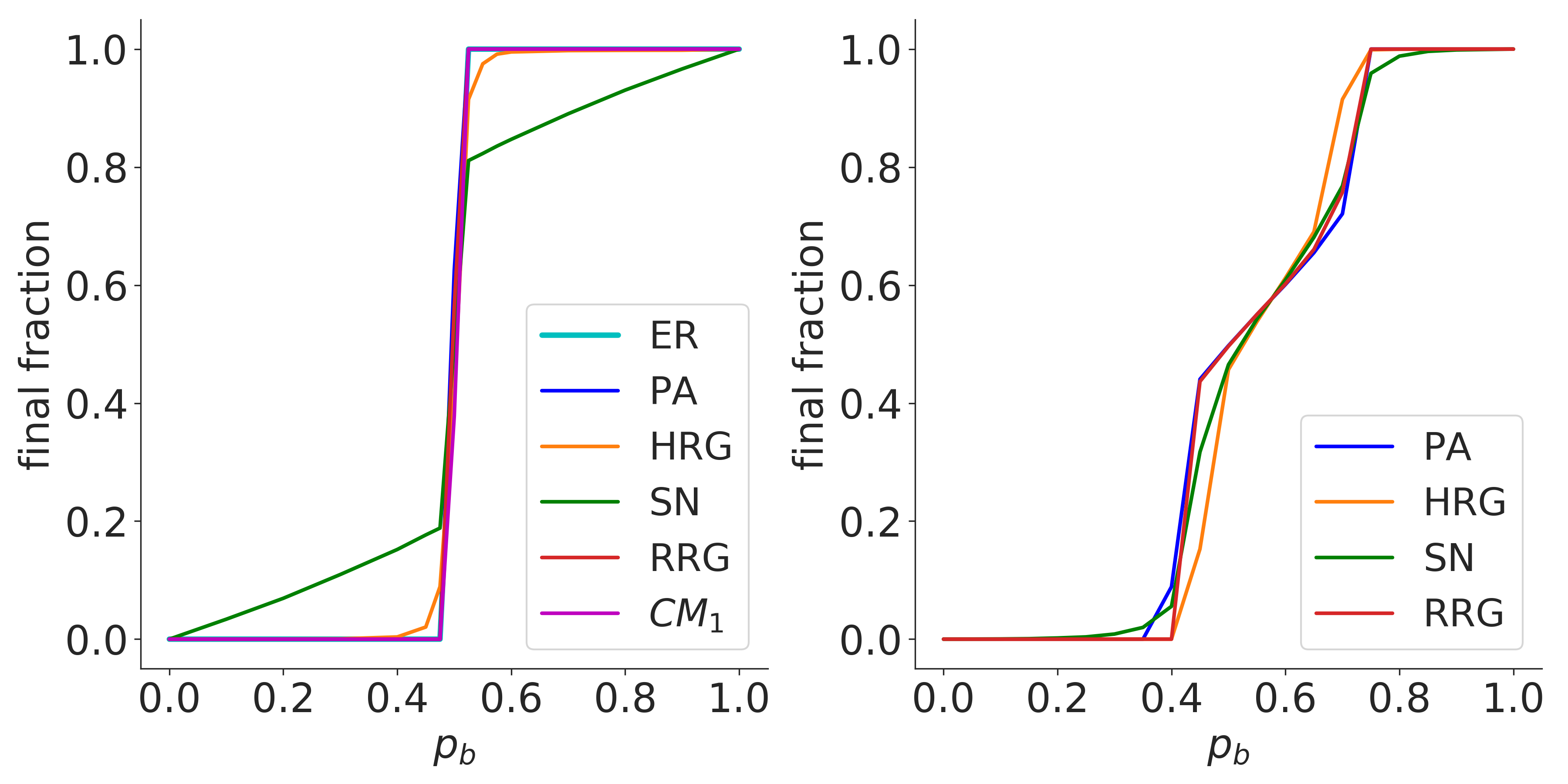}

\caption{The final fraction of black nodes \textbf{(left)} in the majority model with a random initial coloring for different values of $p_b$ on SD SN, several random graphs with comparable parameters and $CM_1$ (which corresponds to the union of SD SN and $\mathcal{G}_{n,d}$ for $d = \bar{d}$, where $\bar{d}$ is the average degree of SD SN), \textbf{(right)} in the $(\psi_1, \psi_2)$-majority model for $\psi_1 = 0.7$ and $\psi_2 = 0.8$ and different values of $p_b$ on TW SN and random graphs PA, HRG and RRG with comparable parameters.}

\label{fig: density12}
\end{figure}
\section{Stabilization Time and Period}

\paragraph{Stabilization in Majority Model.} 
As discussed, prior work has shown that the stabilization time and period of the majority model are bounded by $\mathcal{O}(|E|)$ and 2. However, in the random setting a poly-logarithmic bound on the expected stabilization time is believed to exist (but only proven for a some special classes of graphs).
We provide evidence to support this conjecture. Firstly, we prove in Theorem~\ref{Stabilization time cycle}that the expected stabilization time of the majority model on a cycle $C_n$ is at most $\log n$.
\begin{theorem}
\label{Stabilization time cycle}
Consider the majority model on a cycle $C_n$. If each node is initially black independently with probability $p_b$, then the process stabilizes in $\log n$ rounds a.a.s. 
\end{theorem}

\begin{proof}

Let $u_1,\cdots,u_k$ be a path of length $k \geq 2$. We say that this is a \emph{black path} (resp. \emph{white path}) if all $u_i$s are black (resp. white). We call the path \emph{alternating} if every two adjacent nodes have opposite colors.
Consider an arbitrary coloring of cycle $C_n$. We can partition $C_n$ into black, white and alternating paths. We observe that all nodes in the black and white paths never change their color under the majority model, and the size of an alternating path reduces by at least two in each round until it disappears. (A visualization of this process is depicted in Figure~\ref{fig:cycle bichrom paths}.) This implies that the length of a longest alternating path divided by two is an upper bound on the stabilization time of the process.

\begin{figure}[h]
\centering
\includegraphics[scale=1]{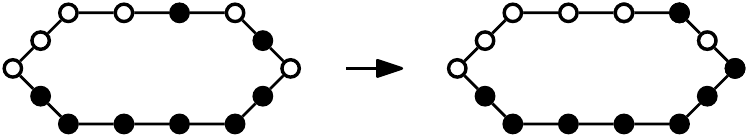}
\caption{A cycle with a black path of length 6, a white path of length 4, and an alternating path of length 4. After one application of the majority model, the length of the alternating path decreases by 2.} 
\label{fig:cycle bichrom paths}
\end{figure}

It thus suffices to show that if we color each node black independently with probability $p_b$, then a.a.s there is no alternating path of size larger than $2 \log n$. Consider an arbitrary path $\mathcal{P}$ of size $2 \log n$ in $C_n$. 
\begin{align*}
  \textrm{Pr}\left[\mathcal{P} \text{ is alternating}\right] = 2\left(p\left(1-p\right)\right)^{\log n}\le
  2\left(\frac{1}{4}\right)^{\log n} = 2n^{-2}
\end{align*}
There are linearly many paths of length $2 \log n$ in $C_n$. Therefore, by a union bound argument, there is no alternating path of size $2 \log n$ a.a.s.
\end{proof}

Furthermore, we investigate the expected stabilization time of the majority model on several real world SNs and random graph models. This is depicted for SD SN and random graphs with comparable parameters in Figure~\ref{figure3}~(left). (See appendix \ref{Sec: Experiment 4} for other SNs.) In our experiments on all these graphs the process ends in less than 30 rounds, while the number of edges is around $m=500000$. Thus, loosely speaking, the expected stabilization time here seems to be poly-logarithmic in $m$ rather than linear. Furthermore, we observe that adding a RRG on top of SD SN (similar to our first countermeasure in Section~\ref{Power-of-elites}) speeds up the process. Hence, adding random edges not only helps the color with initial majority to win, but also makes this happen faster.

\paragraph{Stabilization in ($\psi_1$, $\psi_2$)-Majority Model.} Consider the $(\psi_1$, $\psi_2$)-majority model, for $\psi_1=\psi_2$, on a graph $G=(V,E)$. Building upon a potential function argument, we prove in Theorem ~\ref{upper bound stab time gen} that the stabilization time and period of the process are upper-bounded by $4m^*$ and 2, respectively, where $m^*$ denotes the number of bichromatic edges in the initial coloring. (Recall that an edge is bichromatic if its endpoints have opposite colors.) Note that $m^*\le |E|$ and the ($\psi$, $\psi$)-majority model coincides with the majority model when $\psi$ is slightly larger than $1/2$. Thus, Theorem~\ref{upper bound stab time gen}, as a special case, bounds the stabilization time of the majority model with $\mathcal{O}(|E|)$, previously proven by~\cite{poljak1986pre}. The main idea of our proof is to establish a relation between the $(\psi, \psi)$-majority model on $G$ and a process called the periodic majority model on a weighed graph $H$, which is constructed from $G$. Then, we argue that in this new process the summation of weights of bichromatic edges decreases in each round. The proof is given in Appendix \ref{appendix-stabilization-time-lower}.

\begin{theorem}
\label{upper bound stab time gen}
Consider the $(\psi, \psi)$-majority model, for some $\psi \in (1/2,1]$, on a graph $G = (V,E)$. The stabilization time is upper-bounded by $4m^*$, where $m^*$ is the number of bichromatic edges in the initial coloring, and the period is always one or two. 
\end{theorem}

Furthermore, we experimentally analyze the expected stabilization time of the ($\psi_1, \psi_2$)-majority model. See Figure~\ref{figure3}~(right) for our results on TW SN (and random graph models with comparable parameters), and Appendix \ref{Sec: Experiment 5} for other SN's.We observe that the initial probabilities $p_b$ for which the process takes the longest to stabilize, are identical to the empirically observed threshold values depicted in Figure~\ref{fig: density12}~(right). We note that this is also the case in the majority model, as the stabilization time peaks at $p_b = 1/2$ (visible in Figure~\ref{figure3}~(left)), at which also the phase transition occurs (visible in Figure~\ref{fig: density12}~(left)). 

\begin{figure}[h]
\centering
\includegraphics[scale=0.27]{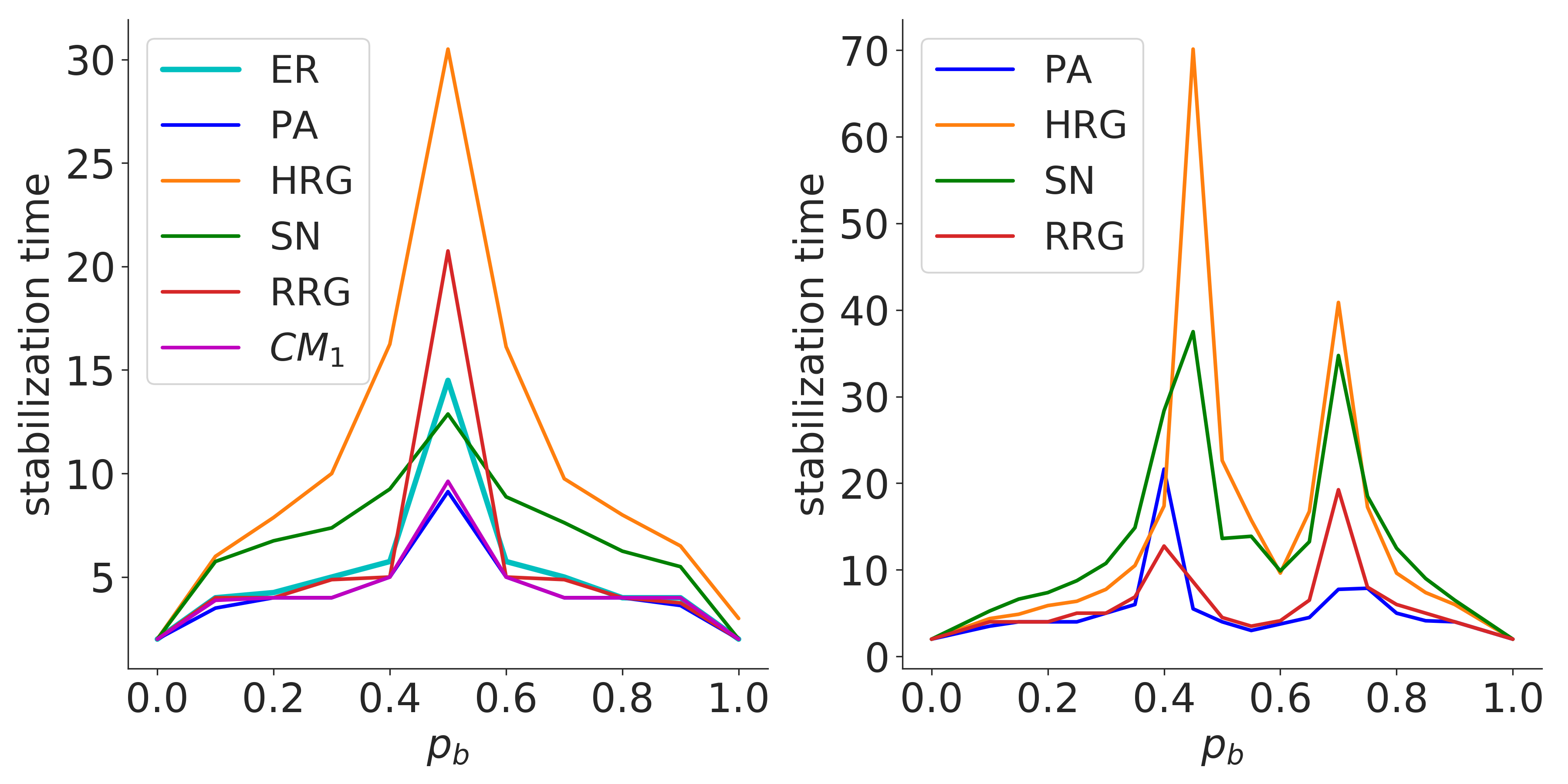}

\caption{The expected stabilization time for different values of $p_b$ in \textbf{(left)} the majority model on SD SN and random graphs with comparable parameters and $CM_1$ (which corresponds to the union of SD SN and $\mathcal{G}_{n,d}$ for $d = \bar{d}$, where $\bar{d}$ is the average degree of SD SN), \textbf{(right)} the $(\psi_1, \psi_2)$-majority model, for $\psi_1 = 0.7$ and $\psi_2 = 0.8$, on TW SN and PA, HRG and RRG with comparable parameters.}

\label{figure3}
\end{figure}

\section{Conclusion}
We showed that in the real world social networks, an extremely small set of high-degree nodes (i.e., elites) can determine the output of an opinion forming process. We developed two countermeasures which can be applied to overpower such a small set of elite nodes. In general, motivated from the study of effective marketing strategies, the problem of finding a small set of agents whose opinion governs the final dominant opinion has been extensively studied for various models. However, the current understanding of the development of countermeasure mechanisms to subdue such a small  group with disproportionate influencing power is limited. We aspire this work to be a starting point for further investigation of effective countermeasures for a large spectrum of models. 

Furthermore, we proved that the ($\psi_1$, $\psi_2$)-majority model exhibits a threshold behavior with two phase transitions at $1-\psi_1$ and $\psi_2$ when the underlying graph is a dense Erd\H{o}s-R\'enyi random graph. Our experiments suggest that a similar threshold behavior might exist for other classes of graphs, but the threshold values are different from $1-\psi_1$ and $\psi_2$. It would be interesting to determine these values in the future work.

In addition, we provided several experimental and theoretical evidence to support the widely believed conjecture of a poly-logarithmic upper bound on the expected stabilization time of the majority model; however, it remains open in its full generality. For the ($\psi_1, \psi_2$)-majority model, for $\psi_1 = \psi_2$, we proved that the stabilization time is bounded by 4 times the number of bichromatic edges in the initial coloring. We believe this can be an important milestone to settle the conjecture. Specifically, if one can prove that from a random coloring the process reaches a coloring with poly-logarithmically many bichromatic edges in a poly-logarithmic number of rounds in expectation, then our result yields the conjecture. 

\bibliographystyle{unsrtnat}
\bibliography{references}

\newpage
\appendix

\setcounter{theorem}{2}
\setcounter{lemma}{0}
\renewcommand{\thetheorem}{\Alph{section}.\arabic{theorem}}
\renewcommand{\thelemma}{\Alph{section}.\arabic{lemma}}

\section{Two Phase Transitions in ($\psi_1$, $\psi_2$)-Majority Model on ER}
\label{Sec: Two Phase Transition}
Consider the ($\psi_1$, $\psi_2$)-majority model on a graph $G=(V,E)$ with a random initial coloring, where each node is colored black independently with probability $p_b$. In Theorem~\ref{threshold-theorem}, we prove that if the minimum degree in $G$ is at least $C\log n$ for a sufficiently large constant $C$, the process goes through two phase transitions at $1-\psi_1$ and $\psi_2$. More accurately, for an arbitrarily small $\zeta>0$ a.a.s.: (i) if $p_b\le (1-\zeta)(1-\psi_1)$, then the white color takes over (ii) if $p_b\ge (1+\zeta)(1-\psi_1)$ and $p_b\le (1-\zeta)\psi_2$, none of the two colors will take over and (iii) if $p_b\ge (1+\zeta)\psi_2$, then the black color takes over. Theorem~\ref{threshold-theorem}, in particular, holds a.a.s. for the Erd\H{o}s-R\'{e}nyi random graph $\mathcal{G}_{n,q}$ when $q\ge \frac{2C\log n}{n}$. This is true because it is well-known that a.a.s. the minimum degree of $\mathcal{G}_{n,q}$ is at least $C\log n$ in this setting, cf.\cite{dubhashi2009concentration}. (This is essentially implied by an application of the Chernoff bound (see Lemma~\ref{Chernoff}) and a union bound.) 

To prove Theorem~\ref{threshold-theorem}, we first need to provide Lemma~\ref{threshold-lemma}, which is built on the Chernoff bound given in Lemma~\ref{Chernoff}.

\begin{lemma}[Chernoff bound, cf.~\cite{dubhashi2009concentration}]
\label{Chernoff}
Suppose that $x_1,\cdots,x_n$ are independent Bernoulli random variables and let $X$ denote their sum, then for $0\leq \epsilon\leq 1$
\begin{itemize}
\item[(i)]$\textrm{Pr}[X\leq \left(1-\epsilon\right)\mathbb{E}[X]]\leq \exp\left({-\frac{\epsilon^2\mathbb{E}[X]}{2}}\right)$
\item[(ii)]$\textrm{Pr}[\left(1+\epsilon\right)\mathbb{E}[X]\leq X]\leq \exp\left({-\frac{\epsilon^2\mathbb{E}[X]}{3}}\right)$.
\end{itemize}
\end{lemma}

\begin{lemma}
\label{threshold-lemma}
Consider an arbitrary node $v$ of degree $d(v)$ in a graph $G=(V,E)$. Assume that we color each node among the neighbors of $v$ with color $a\in \{b,w\}$ independently with some probability $p\in [1/2,1]$. At least $(1-\epsilon) pd(v)$ (analogously, at most $(1+\epsilon)pd(v)$) number of its neighbors will be colored with color $a$ for an arbitrarily small constant $\epsilon>0$ with probability at least $1-n^{-2}$, if $d(v)\ge C\log n$ for a sufficiently large constant $C$.
\end{lemma}
\begin{proof}
We prove that at least $(1-\epsilon) pd(v)$ number of its neighbors will be colored with color $a$ with probability at least $1-n^{-2}$, using Lemma~\ref{Chernoff} (i). (The other case can be proven analogously using Lemma~\ref{Chernoff} (ii).) Let us label the neighbors of $v$ from $u_1$ to $u_d$, where $d:=d(v)$. We define the Bernoulli random variable $x_i$ to be 1 if and only if node $u_i$ is colored with color $a$. Let $X:=\sum_{i=1}^{d}x_i$ be the number of neighbors of $v$ which are colored with $a$. Then, we have $\mathbb{E}[X]=\sum_{i=1}^{d}\textrm{Pr}[x_i=1]=pd$. Since $x_i$s are independent, we can apply Lemma~\ref{Chernoff} (i), which yields
\begin{align*}
  \textrm{Pr}[X\le (1-\epsilon)pd]&=\textrm{Pr}[X\le (1-\epsilon)\mathbb{E}[X]] \\ & \le \exp\left(-\frac{\epsilon^2\mathbb{E}[X]}{2}\right) \\ & = \exp\left(-\frac{\epsilon^2pd}{2}\right) \\ & \le \exp\left(-\frac{\epsilon^2C\log n}{4}\right)
\end{align*}
where in the last step we used $d\ge C\log n$ and $p\ge 1/2$. For a sufficiently large constant $C$, this probability is bounded by $\exp(-2\log n)=n^{-2}$. Thus, we have $X\ge (1-\epsilon)pd$ with probability at least $1-n^{-2}$.
\end{proof}

\begin{theorem}
\label{threshold-theorem}
Consider the ($\psi_1$, $\psi_2$)-majority model on a graph $G=(V,E)$ for some $\psi_1,\psi_2\in (1/2,1)$ and assume that each node is black independently with probability $p_b$. Suppose that all nodes in $G$ are of degree at least $C\log n$ for a sufficiently large constant $C$. Then, for an arbitrarily small constant $\zeta>0$, we have a.a.s. that
\begin{itemize}
\item[(i)] if $p_b\le (1-\zeta)(1-\psi_1)$, the white color takes over
\item[(ii)] if $p_b\ge (1+\zeta)(1-\psi_1)$ and $p_b\le (1-\zeta)\psi_2$, none of the colors will take over
\item[(iii)] if $p_b\ge (1+\zeta)\psi_2$, the black color takes over.
\end{itemize}
\end{theorem}
\begin{proof}
We prove that if $p_b\ge (1+\zeta)\psi_2$, then the black color takes over a.a.s. and it does not if $p_b\le (1-\zeta)\psi_2$. One can analogously prove that if $p_b\le (1-\zeta) (1-\psi_1)$, then the white color takes over a.a.s. and it does not if $p_b\ge (1+\zeta)(1-\psi_1)$.

Let us first consider the setting of $p_b\ge (1+\zeta)\psi_2$. Applying Lemma~\ref{threshold-lemma} and a union bound implies that for every node $v$, at least $(1-\epsilon)p_bd(v)\ge (1-\epsilon)(1+\zeta)\psi_2d(v)$ of its $d(v)$ neighbors are black with probability at least $1-n^{-1}$. For a sufficiently small choice of $\epsilon>0$, we have $(1-\epsilon) (1+\zeta)\psi_2d(v)> \psi_2 d(v)$. Therefore, with probability at least $1-n^{-1}$ more than $\psi_2$ fraction of the neighbors of every node is black. This implies that a.a.s. all nodes will be colored black after one round.

Now, let us prove that if $p_b \le (1-\zeta) \psi_2$, then a.a.s. the black color does not take over. By a simple application of the Chernoff bound (Lemma~\ref{Chernoff}), one can show that a.a.s. the number of white nodes in the initial coloring is not zero.
Furthermore based on Lemma~\ref{threshold-lemma} and a union bound, with probability at least $1-n^{-1}$ for every white node $v$ at most $(1+\epsilon) p_b d(v)\le (1+\epsilon)(1-\zeta)\psi_2 d(v)$ of its $d(v)$ neighbors are black. Again, by selecting $\epsilon>0$ to be sufficiently small we have $(1+\epsilon)(1-\zeta)\psi_2d(v)<\psi_2 d(v)$. Therefore, with probability at least $1-n^{-1}$, all the white nodes remain white forever, which implies that the black color does not take over.
\end{proof}

\section{Proof of Theorem~\ref{upper bound stab time gen}}
\label{appendix-stabilization-time-lower}
\begin{customthm}{\ref{upper bound stab time gen}}
Consider the $(\psi, \psi)$-majority model, for some $\psi \in (1/2,1]$, on a graph $G = (V,E)$. The stabilization time is upper-bounded by $4m^*$, where $m^*$ is the number of bichromatic edges in the initial coloring, and the period is always one or two.
\end{customthm}
\begin{proof}
Assume that $V = \{v_1,..,v_n\}$. Let $H = (V^H, E^H, \omega)$ be a weighted bipartite graph, which is built based on $G$ (the original graph) in the following way. We define $V^H := X \cup Y$ for $X := \{x_1,...,x_n\}$ and $Y := \{y_1,...,y_n\}$. Furthermore, we let $E^H := \{\{x_i,y_j\} : \{v_i, v_j\} \in E \text{ or } i=j\}$. We define an edge $\{x_i, y_i\}$ to have weight $2 \psi d(v_i) - d(v_i) - \frac{1}{2}$ if $\psi d(v)$ is an integer, and weight $2 \lfloor \psi d(v_i) \rfloor - d(v_i) +1-\frac{1}{4n}$ otherwise. All other edges have weight 1. 

Consider an arbitrary initial coloring of $G$. We  then color $H$ in the following way. We let both $x_i$ and $y_i$ for $1 \leq i \leq n$ have the same color as $v_i$ upon initialization. Now, assume that we run the ($\psi$, $\psi$)-majority model on $G$, but we run a different process called the \emph{periodic majority model} on $H$. In the periodic majority model in the $t$-th round, if $t$ is odd, all nodes in $X$ update their color, and if $t$ is even, all nodes in $Y$ update their color. When a node updates, it chooses (similarly as in the majority model) the \textit{weighted majority} among its neighbors. Important to note is that in this setting $y_i$ (resp. $x_i$) is a neighbor of $x_i$ (resp. $y_i$). We observe that a tie is never possible in this setting because all the incident edges of a node have weight 1, except for exactly one edge.
 
Consider an arbitrary initial coloring on $G$ and let $\mathcal{C}_t(v_i)$, for $1 \leq i \leq n$ and $t \in \mathbb{N}_0$, denote the color of node $v_i$ after $t$ rounds of the ($\psi$, $\psi$)-majority model. Similarly, let $\mathcal{C}_t(x_i)$ and $\mathcal{C}_t(y_i)$ denote the colors of nodes $x_i$ and $y_i$ after $t \in \mathbb{N}_0$ rounds of the periodic majority model on $H$. We claim that for odd $t$, $\mathcal{C}_t(v_i) = \mathcal{C}_t(x_i)$ and for even $t$, $\mathcal{C}_t(v_i) = \mathcal{C}_t(y_i)$ for all $1 \leq i \leq n$. We prove this claim by applying induction on $t$. For the base case, at $t=0$, we note that the claim holds because we know that $\mathcal{C}_0(v_i)=\mathcal{C}_0(y_i)$ based on our color initialization. For the inductive step, we assume that the statement is true for some $t$ and prove that it is also true for $t+1$. Suppose that $t$ is even (the proof holds in the same way when $t$ is odd). By the induction hypothesis, $\mathcal{C}_t(y_i) = \mathcal{C}_t(v_i)$ for $1 \leq i \leq n$. Consider an arbitrary node $v_j$, $1 \leq j \leq n$. Without loss of generality, assume that $\mathcal{C}_t(v_j)=w$. We prove that $\mathcal{C}_{t+1}(v_j) = \mathcal{C}_{t+1}(x_j)$, meaning that $\mathcal{C}_{t+1}(v_j) = b$ if and only if $\mathcal{C}_{t+1}(x_j) = b$. Let $n_b^t(v_j)$ denote the number of black neighbor of $v_j$ in the $t$-th round and assume that $n_b^t(x_j)$ denotes the number of $y_i$s for $i \neq j$ in $x_j$'s neighborhood which are black in the $t$-th round. Note that by construction and the induction hypothesis, $n_b^t(v_j) = n_b^t(x_j)$.
Let us first consider the case that $\psi d(v_j)$ is an integer. We observe that $\mathcal{C}_{t+1}(v_j) = b$ if and only if $n^t_b(v_j) \geq \psi d(v_j)$. We have that

\begin{align*}
  n_b^t(v_j) \geq \psi d(v_j) \overset{n_b^t(v_j) = n_b^t(x_j)}{\iff} n_b^t(x_j) \geq \psi d(v_j) \\
  \overset{\psi d(v_i)\in \mathbb{N}}{\iff} 
  n_b^t(x_j) > \psi d(v_j) - \frac{1}{4} \iff \frac{n_b^t(x_j)}{2\psi d(v_j) -\frac{1}{2}} > \frac{1}{2}
\end{align*}

where we used the fact that for two integers $e$ and $f$, $e\ge f$ is the same as $e>f-1/4$. The summation of the weights of edges incident to $x_j$ is equal to $d(v_j) + (2 \psi d(v_j) - d(v_j) - \frac{1}{2}) = 2 \psi d(v_j) - \frac{1}{2}$. Therefore, 
$\mathcal{C}_{t+1}(x_j) = b$ if and only if $\frac{n_b^t(x_j)}{2 \psi d(v_j) - \frac{1}{2}} > \frac{1}{2}$ (note that $\mathcal{C}_t(y_j)=w$ and ties are not possible in the periodic majority model on $H$). From this it follows that $\mathcal{C}_{t+1}(v_j)$ = $b$ if and only if $\mathcal{C}_{t+1}(x_j)$= $b$. 

Let us now consider the case where $\psi d(v_j)$ is not an integer. We know that $\mathcal{C}_{t+1}(v_j) = b$ if and only if $n_b^t(v_j) \geq \lceil \psi d(v_j) \rceil$, and again $n_b^t(v_j) = n_b^t(x_j)$. We note that
\begin{align*}
  n_b^t(v_j) \geq \lceil \psi d(v_j) \rceil & \overset{n_b^t(v_j) = n_b^t(x_j)}{\iff} n_b^t(x_j) \geq \lceil \psi d(v_j) \rceil \\ &
  \iff
  n_b^t(x_j) > \lfloor \psi d(v_j) \rfloor +\frac{1}{2}- \frac{1}{8n} \\ &\iff \frac{n_b^t(x_j)}{2 \lfloor \psi d(v_j) \rfloor +1-\frac{1}{4n}} > \frac{1}{2}
\end{align*}
where we used that for an integer $e$ and a non-integer $f$, we know that $e\ge \lceil f\rceil$ is the same as $e>\lfloor f\rfloor +1/2-1/8n$ for $n\ge 1$. We observe that the summation of the weights of edges incident to $x_j$ is equal to $d(v_j) + (2 \lfloor \psi d(v_j) \rfloor - d(v_j) +1 - \frac{1}{4n}) = 2 \lfloor \psi d(v_j) \rfloor +1 -\frac{1}{4n}$. Hence, $\mathcal{C}_{t+1}(x_j) = b$ if and only if $\frac{n_b^t(x_j)}{2 \lfloor \psi d(v_j)\rfloor + 1-\frac{1}{4n}} > \frac{1}{2}$ (where we are again using the facts that $\mathcal{C}_t(y_j)=w$ and ties are not possible in the periodic majority model on $H$). Thus also for the non integer case, $\mathcal{C}_{t+1}(v_j)=b$ if and only if $\mathcal{C}_{t+1}(x_j)=b$. 

So far, we have established a relation between the ($\psi, \psi$)-majority model on $G$ and the periodic majority model on $H$. Assume that $m^*$ is the number of bichromatic edges in the initial coloring of the ($\psi$, $\psi$)-majority model on $G$. We prove that the corresponding periodic majority process on $H$, which we described above, will reach a fixed coloring after at most $4m^*$ rounds. Putting this statement in parallel with what we proved before, we can conclude that in the ($\psi, \psi$)-majority model on $G$, after at most $t$ rounds for some $t \leq 4m^*$, $\mathcal{C}_t(v_i) = \mathcal{C}_{t+2}(v_i)$ and $\mathcal{C}_{t+1}(v_i) = \mathcal{C}_{t+3}(v_i)$ for each $1 \leq i \leq n$. In other words, the process reaches a cycle of colorings of length one or two in at most $4m^*$ rounds.

It remains to prove that the corresponding periodic majority process on $H$ reaches a fixed coloring in at most $4m^*$ rounds. Let $\phi_t^1$ denote the summation of the weights of the bichromatic edges and let $\phi_t^2$ be the number of bichromatic edges of the form $\{x_i,y_i\}$ for $1\le i\le n$ in the $t$-th round. We define the potential function $\phi_t:=\phi_t^1+\phi_t^2/2$. We start by making a couple of observations.
\paragraph{Fact 1.} Note that in the initial coloring $\mathcal{C}_0(x_i)=\mathcal{C}_0(y_i)$ for $1\le i\le n$, which implies that $\phi_0^2=0$. Furthermore, for a node $x_i$, the edge $\{x_i,y_i\}$ is monochromatic and all other incident edges are of weight 1. Thus, the summation of weight of the bichromatic edges incident to $x_i$ is simply equal to the number of bichromatic edges incident to $v_i$ in the initial coloring on $G$. A simple counting argument gives us $\phi_0^1=2m^*$. Therefore, we have that $\phi_0=\phi_0^1+\phi_0^2/2=2m^*$.
\paragraph{Fact 2.}We know that $2 \psi d(v_j) - d(v_j) - 1/2 \geq 1/2$ when $\psi d(v_j)$ is an integer and $2 \lfloor \psi d(v_i) \rfloor - d(v_i) +1 - \frac{1}{4n} \geq -\frac{1}{4n}$ since $\psi>1/2$. We note that all edges have positive weight, except for the edges of the form $\{x_i, y_i\}$ for $1 \leq i \leq n$ for which we have just shown that they have weight at least $-\frac{1}{4n}$.
This implies that $\phi_t^1\ge -\frac{n}{4n}=-\frac{1}{4}$ and trivially $\phi_t^2\ge 0$ for any $t \in \mathbb{N}$. Hence $\phi_t\ge -\frac{1}{4}$.
\paragraph{Fact 3.} Consider an arbitrary $t\in \mathbb{N}_0$. If in rounds $t$ and $t+1$ at least one node changes its color, then the value of potential function deceases at least by $1/2$, i.e., $\phi_{t+1}\le \phi_{t-1}-1/2$.

To see the correctness of Fact 3, assume that a node $x_i$ for $1\le i\le n$ changes its color in the $t$-th round, we claim that the value of potential function decreases at least by $1/2$. (The proof for $y_i$ is analogous.) We need to consider two cases of $x_i$ changing from black to white or changing from white to black. We essentially need to apply the following two observations, then one can check by some simple calculations that the edge weights are tailored in such a way to make our claim work. (i) All edges incident to a node $x_i$ are of weight 1, except from the edge $\{x_i,y_i\}$. (ii) Since $x_i$ changes its color, then the summation of weights of the incident edges whose endpoints have the opposite color is strictly larger than one half of the summation of weights of all incident edges. (Recall that a tie is not possible.)

Let us now combine the above three facts. The potential function is equal to $2m^*$ initially. Then, its value decreases at least by $1/2$ in every two consecutive rounds. Since it cannot have a value smaller than $-1/4$, after at most $2(2m^*+1/4)=4m^*+1/2$ rounds, it must reach a fixed coloring. Since the number of rounds and $m^*$ are both integers, we get the slightly stronger bound of $4m^*$. This finishes the proof.
\end{proof}

\clearpage

\section{Further Plots}

\subsection{Winning sets and Coutermeasures}
\label{Sec: Experiment 1}

\begin{figure}[H]
	\centering
	\begin{subfigure}{.3\textwidth}
		\includegraphics[scale=0.40]{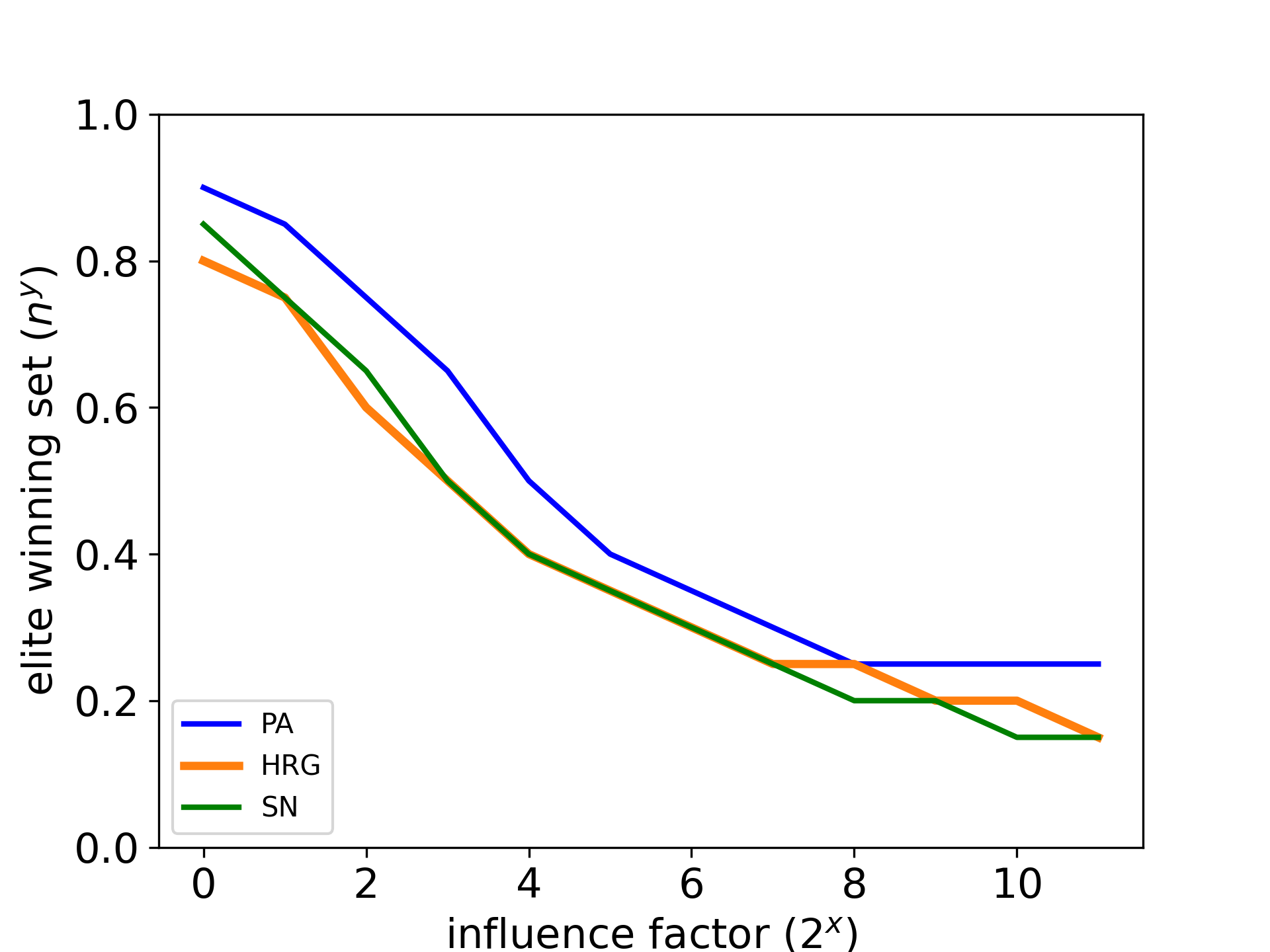}
		\caption{SD SN}
	\end{subfigure}
	\begin{subfigure}{.3\textwidth}
		\includegraphics[scale=0.40]{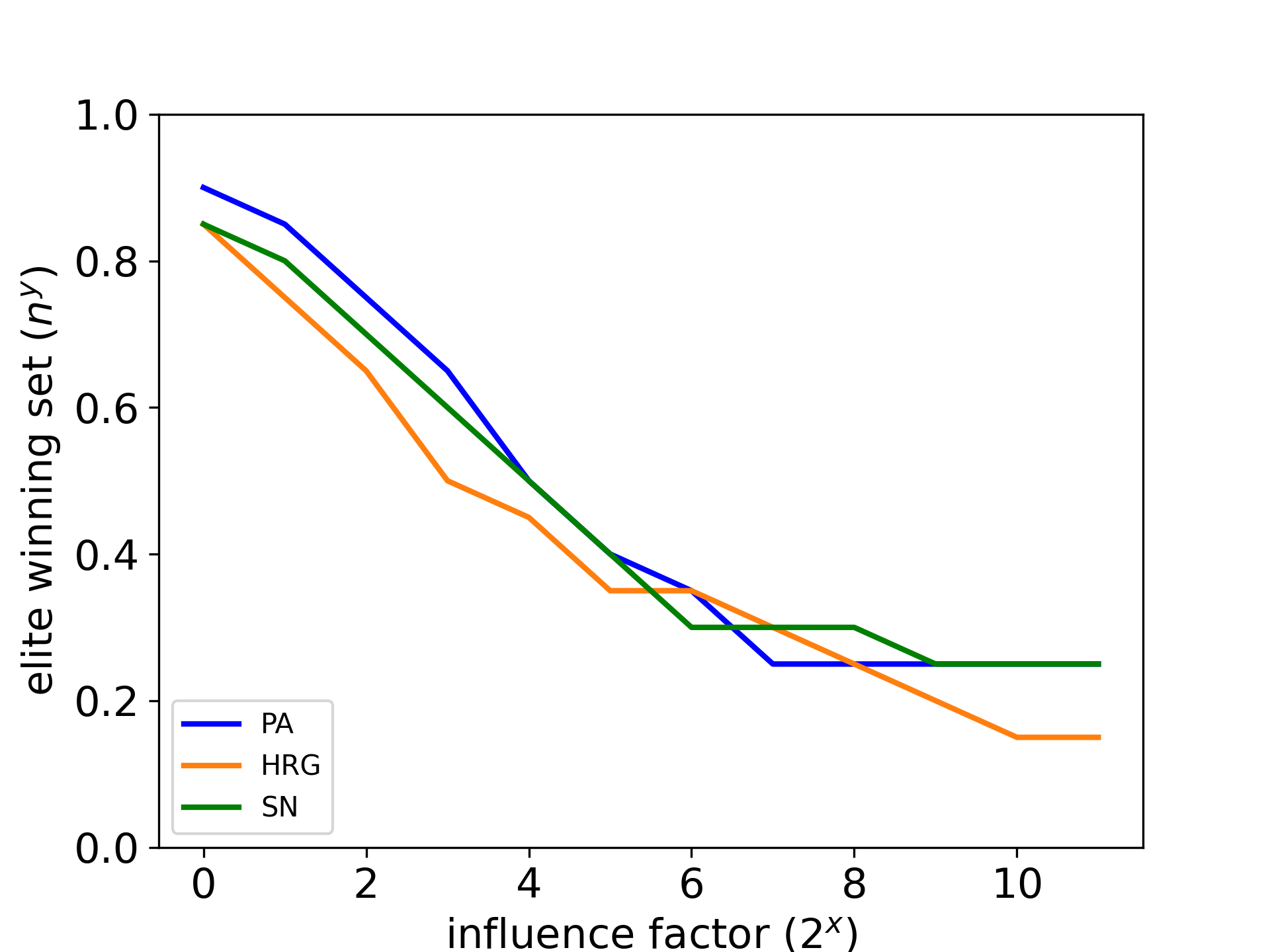}
		\caption{FB SN}
	\end{subfigure}
	\begin{subfigure}{.3\textwidth}
		\includegraphics[scale=0.40]{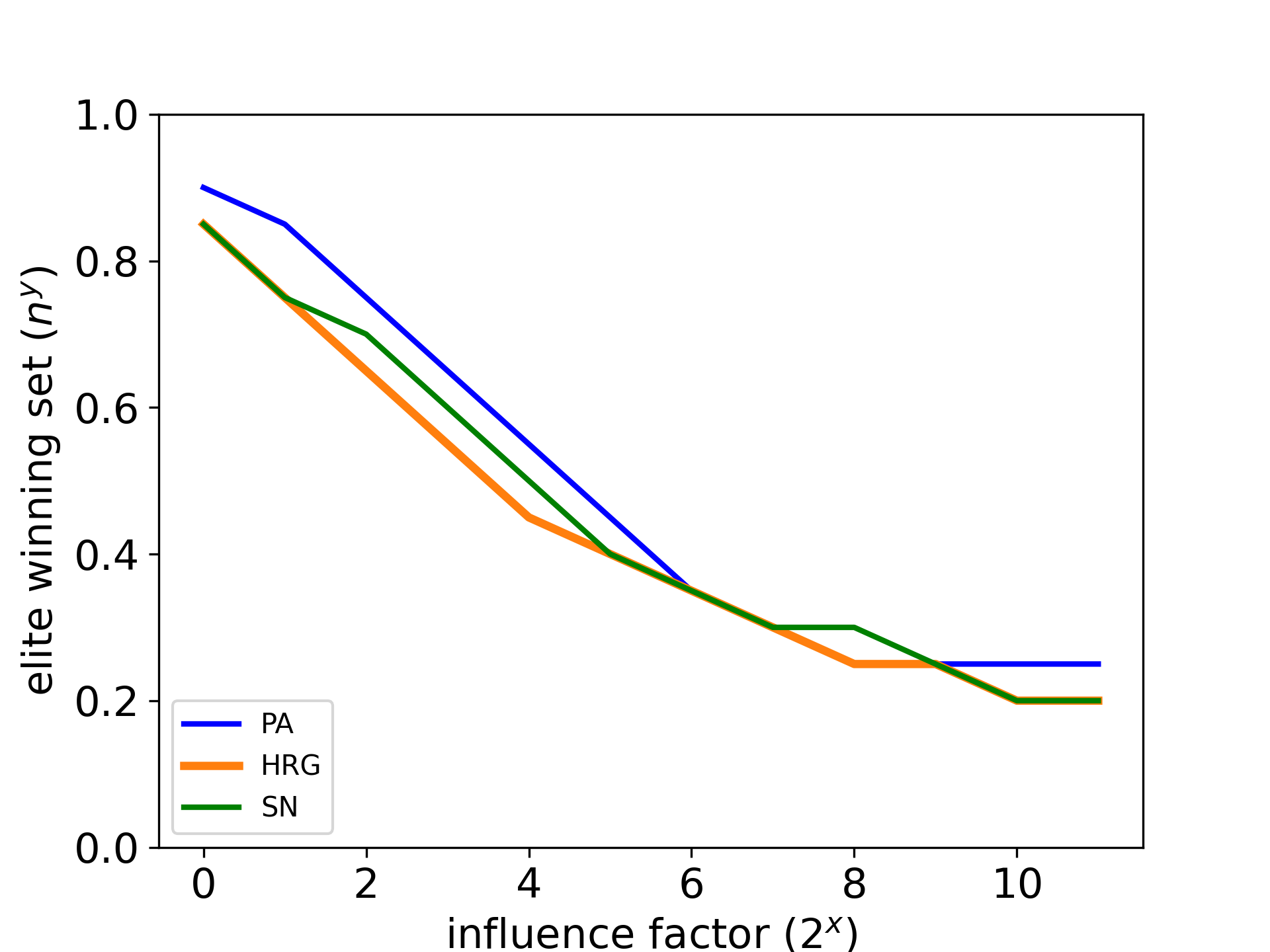}
		\caption{TW SN}
	\end{subfigure}
	\caption{The minimum size of a winning set of elite nodes whose influence factor is $r=2^x$ (while the influence factor of the rest of nodes is 1) \textbf{(a)} on SD SN \textbf{(b)} on FB SN \textbf{(c)} on TW SN, and HRG and PA with comparable parameters.}
\end{figure}


\begin{figure}[H]
	\centering
	\begin{subfigure}{.3\textwidth}
		\includegraphics[scale=0.40]{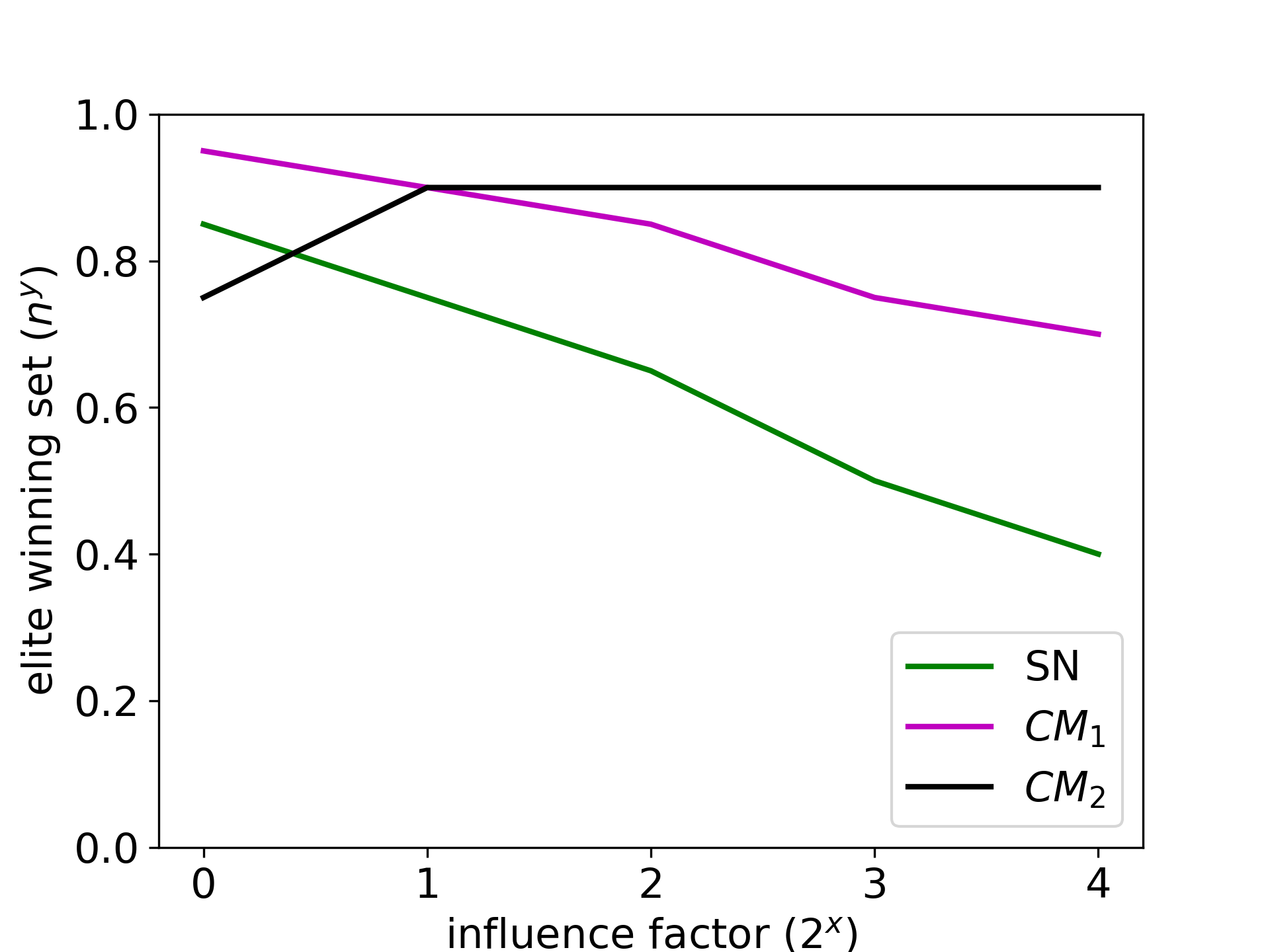}
		\caption{SD SN}
	\end{subfigure}
	\begin{subfigure}{.3\textwidth}
		\includegraphics[scale=0.40]{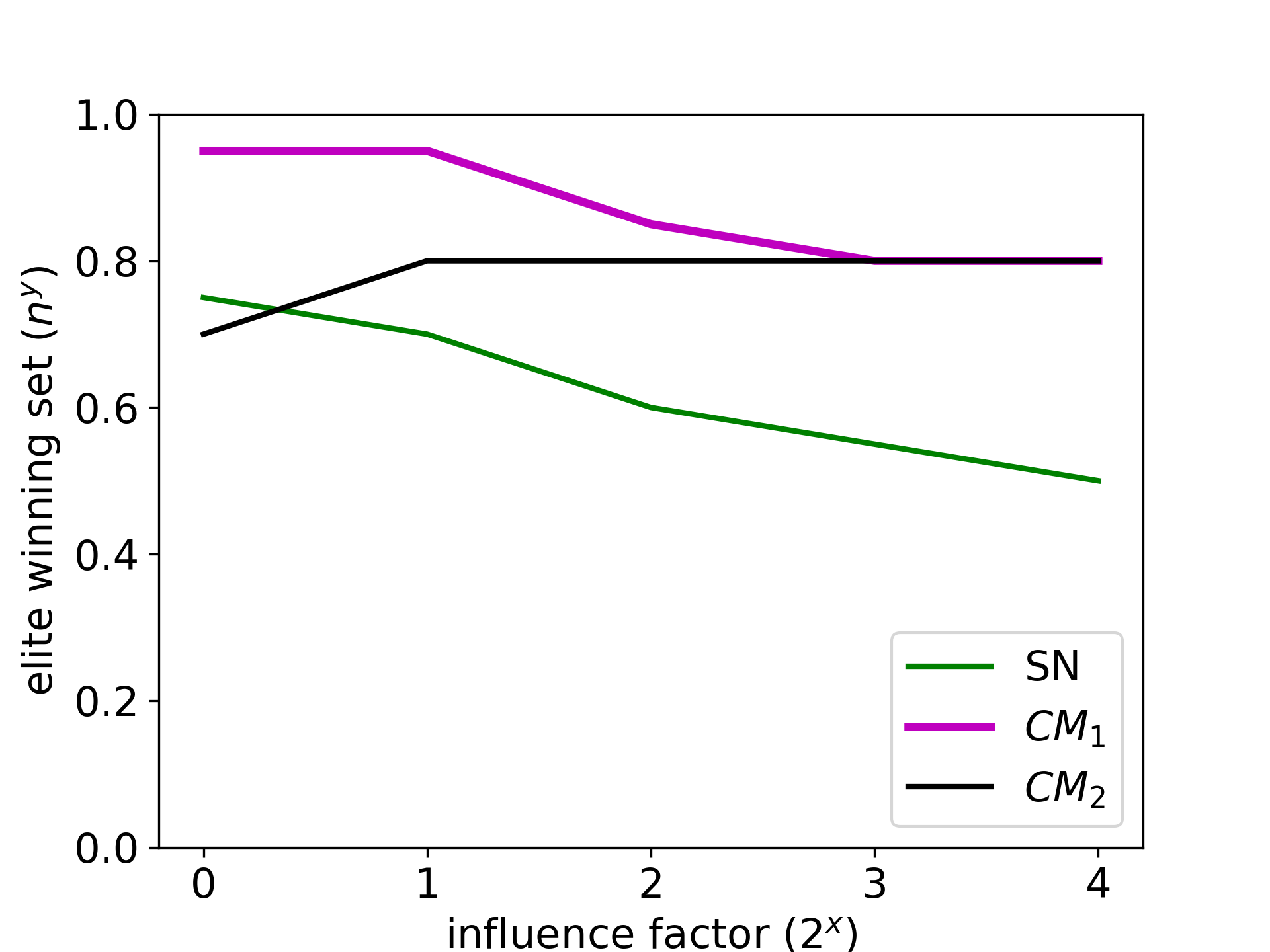}
		\caption{YT SN}
	\end{subfigure}
	\begin{subfigure}{.3\textwidth}
		\includegraphics[scale=0.40]{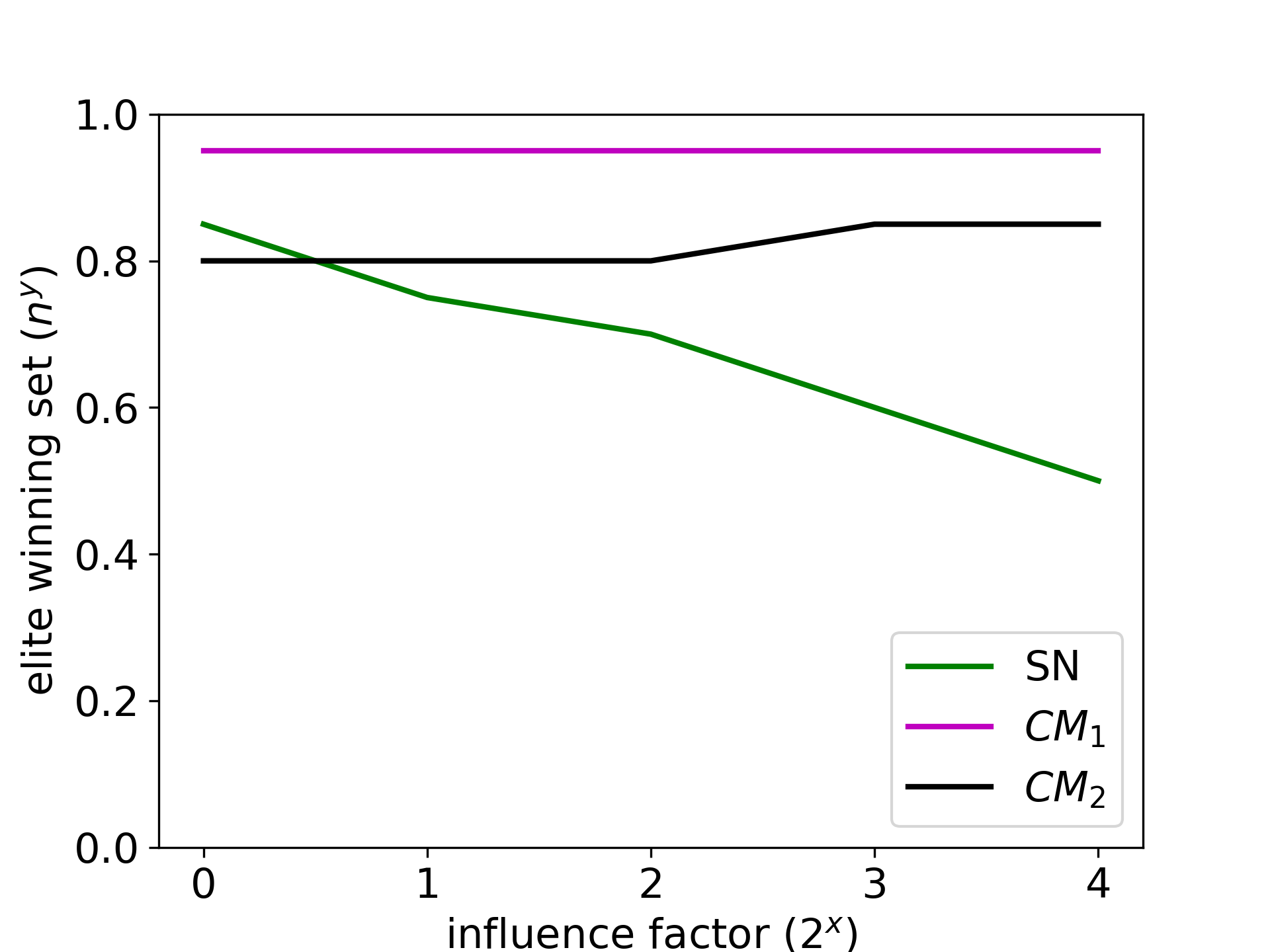}
		\caption{TW SN}
	\end{subfigure}
	\caption{\textbf{(a)} The minimum size of a winning set of elite nodes whose influence factor is $r=2^x$ (while the influence factor of the rest of nodes is 1) on SD SN, and the graphs $CM_1$ and $CM_2$, corresponding to our countermeasures. $CM_1$ is the union of SD SN and a RRG with degree $d=2r\bar{d}$ (where $\bar{d}$ is the average degree of SD SN). $CM_{2}$ denotes SD SN in which all nodes get assigned a stubbornness factor $ \gamma = 1 - 1/2r$. \textbf{(b)} and \textbf{(c)} provide the same plots for YT SN and TW SN, respectively.}
\end{figure}

\clearpage

\subsection{Random Initial Coloring Majority Model}
\label{Sec: Experiment 2}

\begin{figure}[H]
	\centering
	\begin{subfigure}{.3\textwidth}
		\includegraphics[scale=0.40]{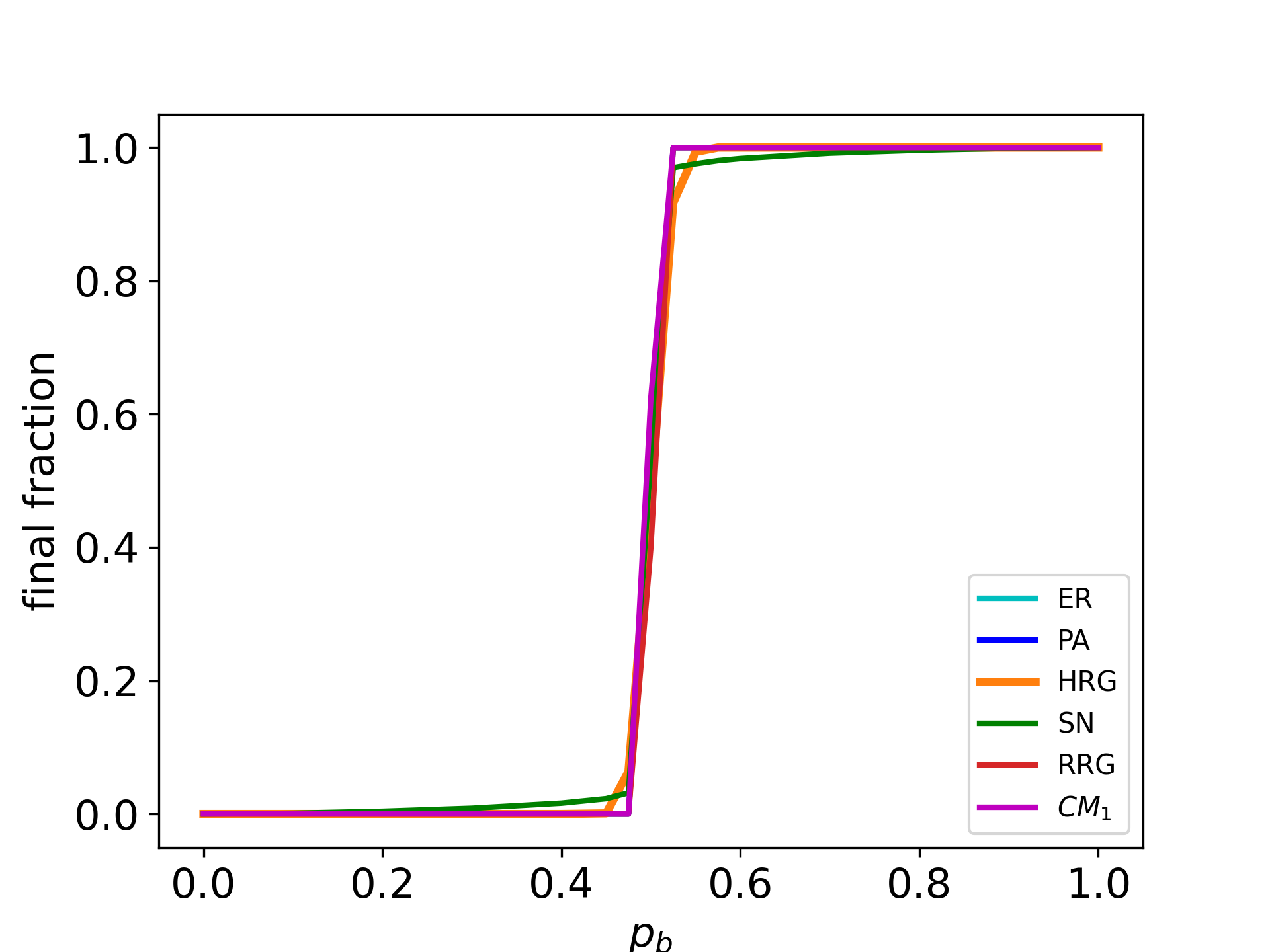}
		\caption{FB SN}
	\end{subfigure}
	\begin{subfigure}{.3\textwidth}
		\includegraphics[scale=0.40]{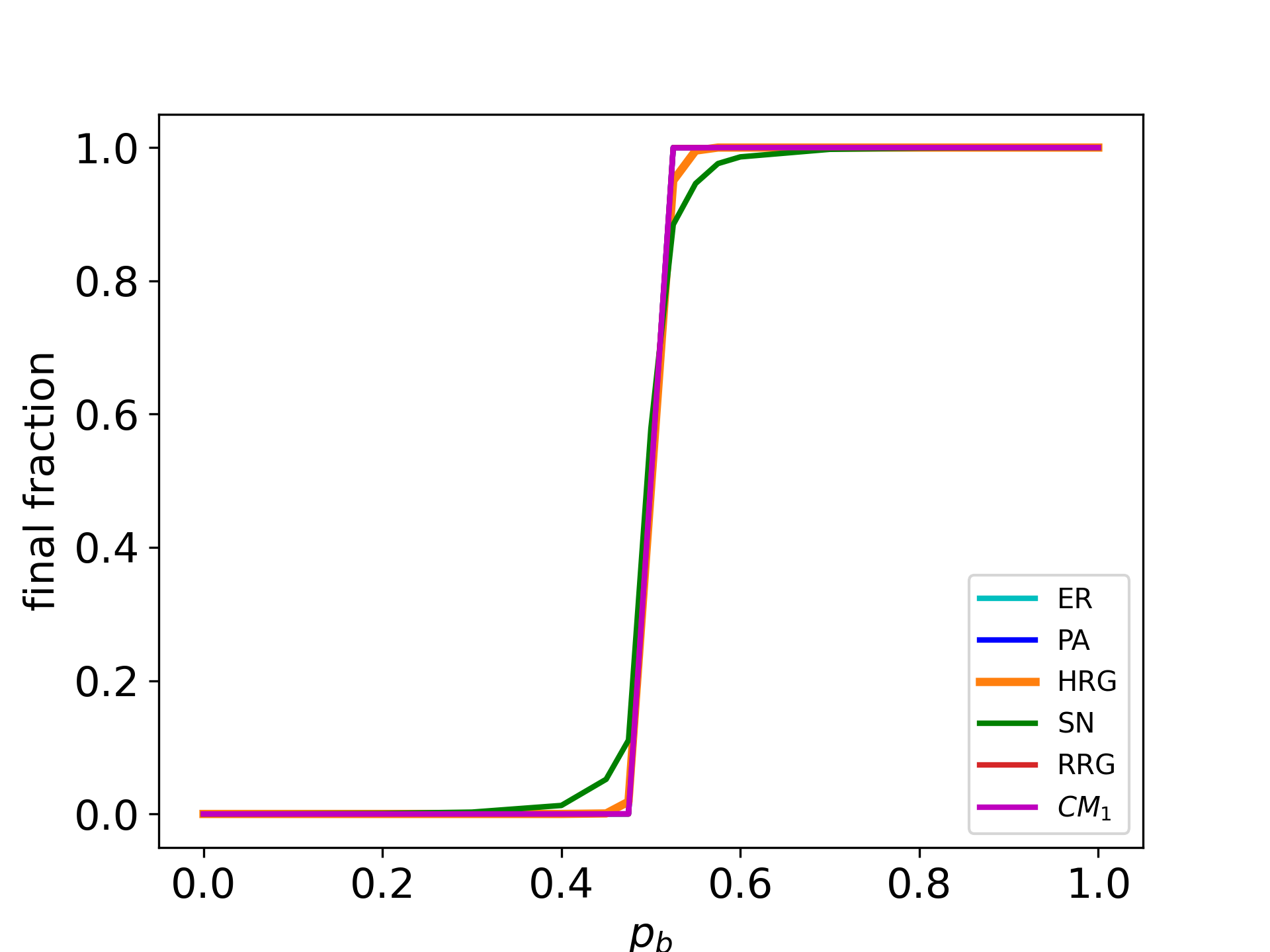}
		\caption{TW SN}
	\end{subfigure}
	\begin{subfigure}{.3\textwidth}
		\includegraphics[scale=0.40]{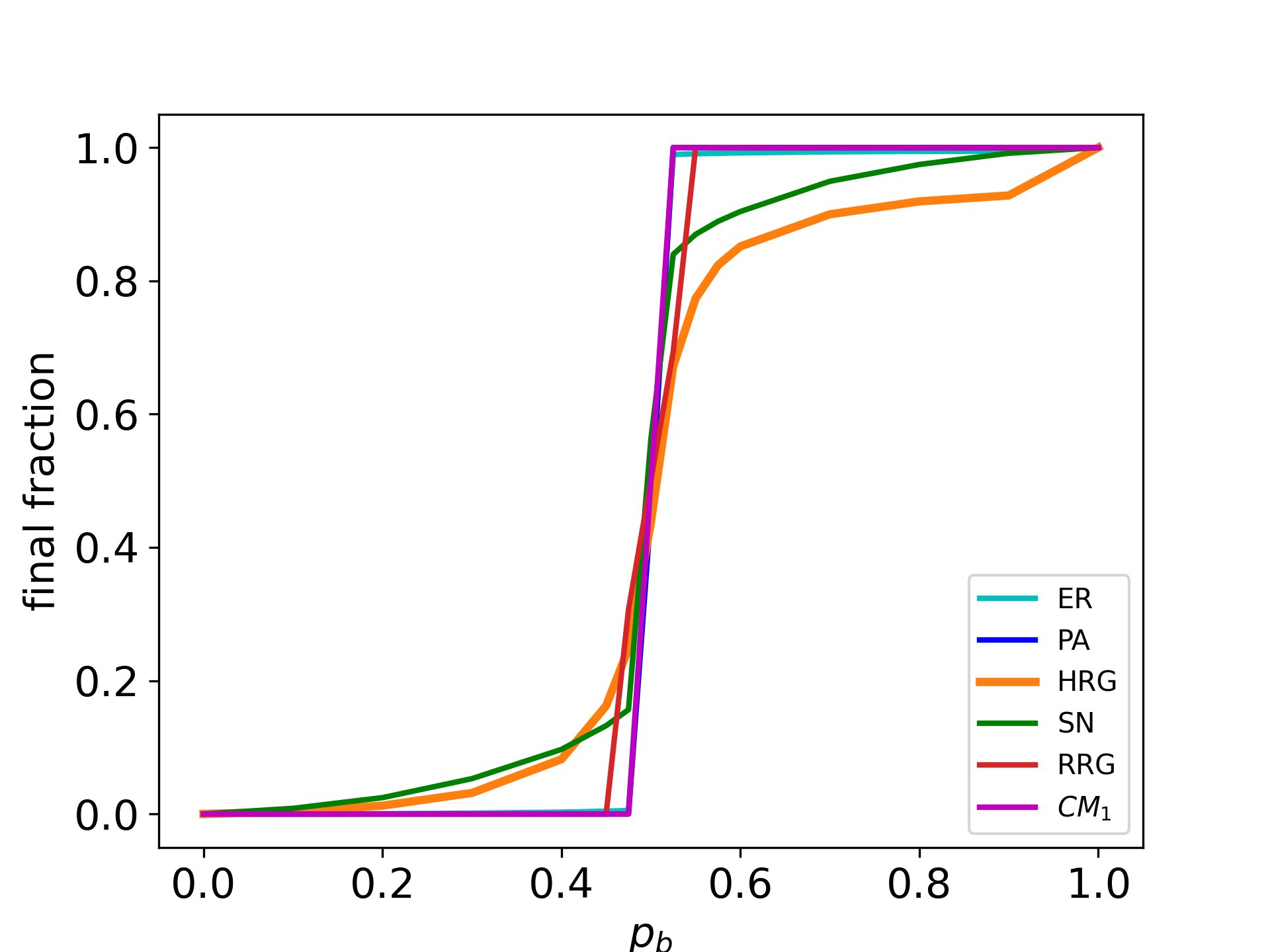}
		\caption{YT SN}
	\end{subfigure}
	\caption{\textbf{(a)} The final fraction of black nodes in the majority model with a random initial coloring on FB SN and ER, PA, HRG, and RRG with comparable parameters. $CM_1$ corresponds to the union of FB SN and $\mathcal{G}_{n,d}$ for $d = \bar{d}$ where $\bar{d}$ is the average degree of FB SN. \textbf{(b)} and \textbf{(c)} provide the same plots for TW SN and YT SN, respectively.}
\end{figure}


\subsection{Random Initial Coloring $(\psi_1, \psi_2)$-Majority Model}
\label{Sec: Experiment 3}
\vspace{8mm}

\begin{figure}[H]
	\centering
	\begin{subfigure}{.3\textwidth}
		\includegraphics[scale=0.40]{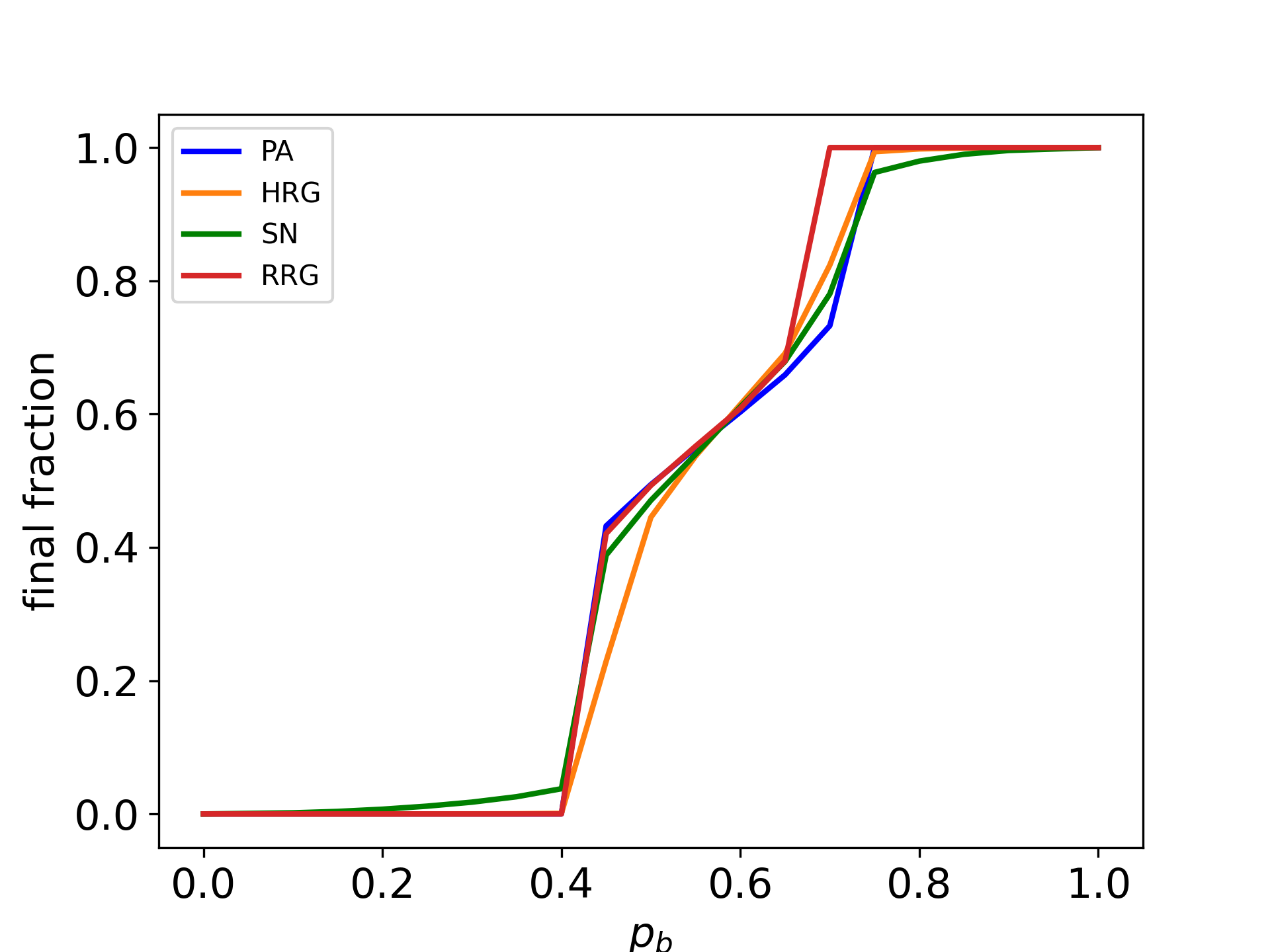}
		\caption{FB SN}
	\end{subfigure}
	\begin{subfigure}{.3\textwidth}
		\includegraphics[scale=0.40]{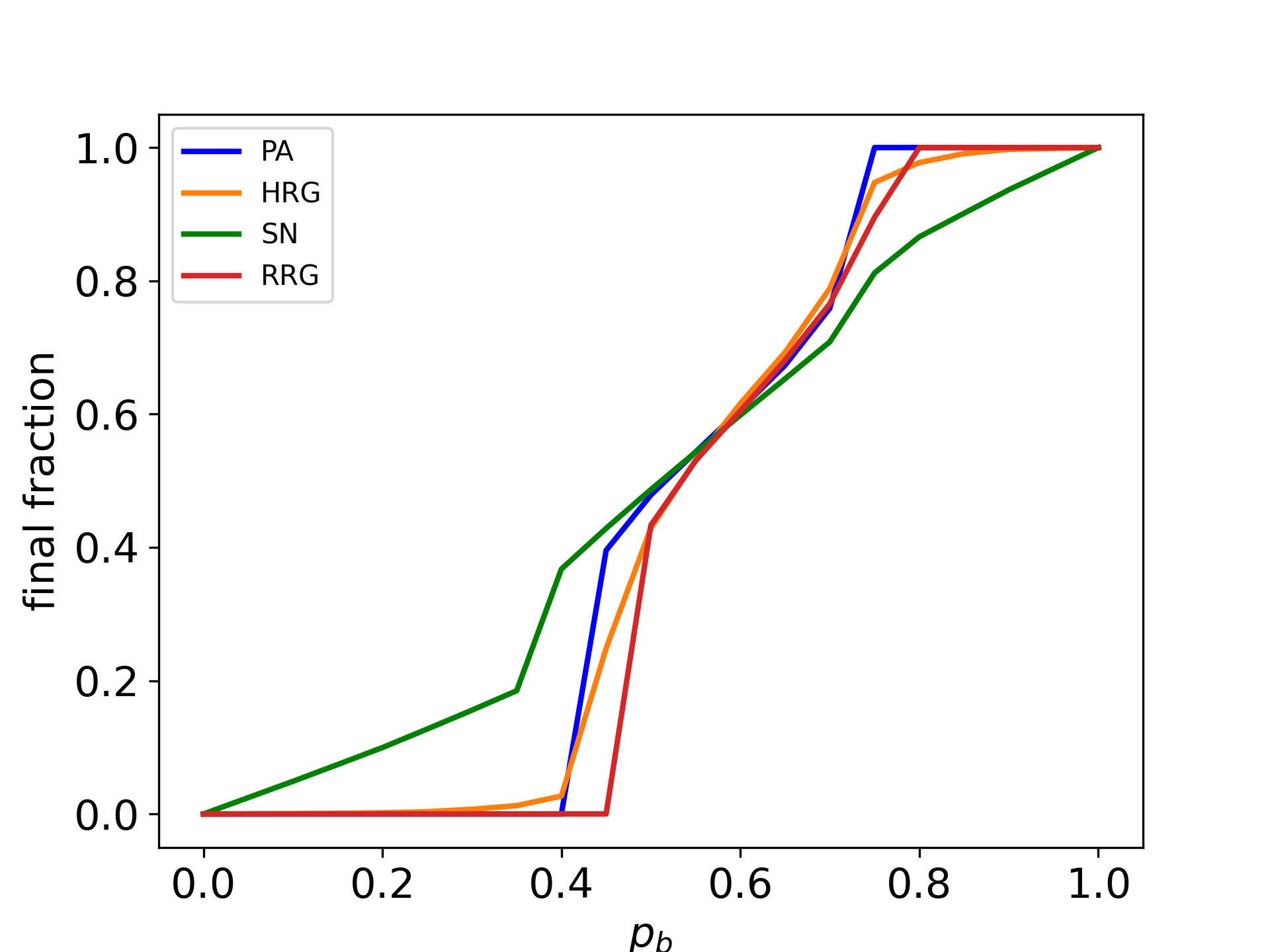}
		\caption{SD SN}
	\end{subfigure}
	\begin{subfigure}{.3\textwidth}
		\includegraphics[scale=0.40]{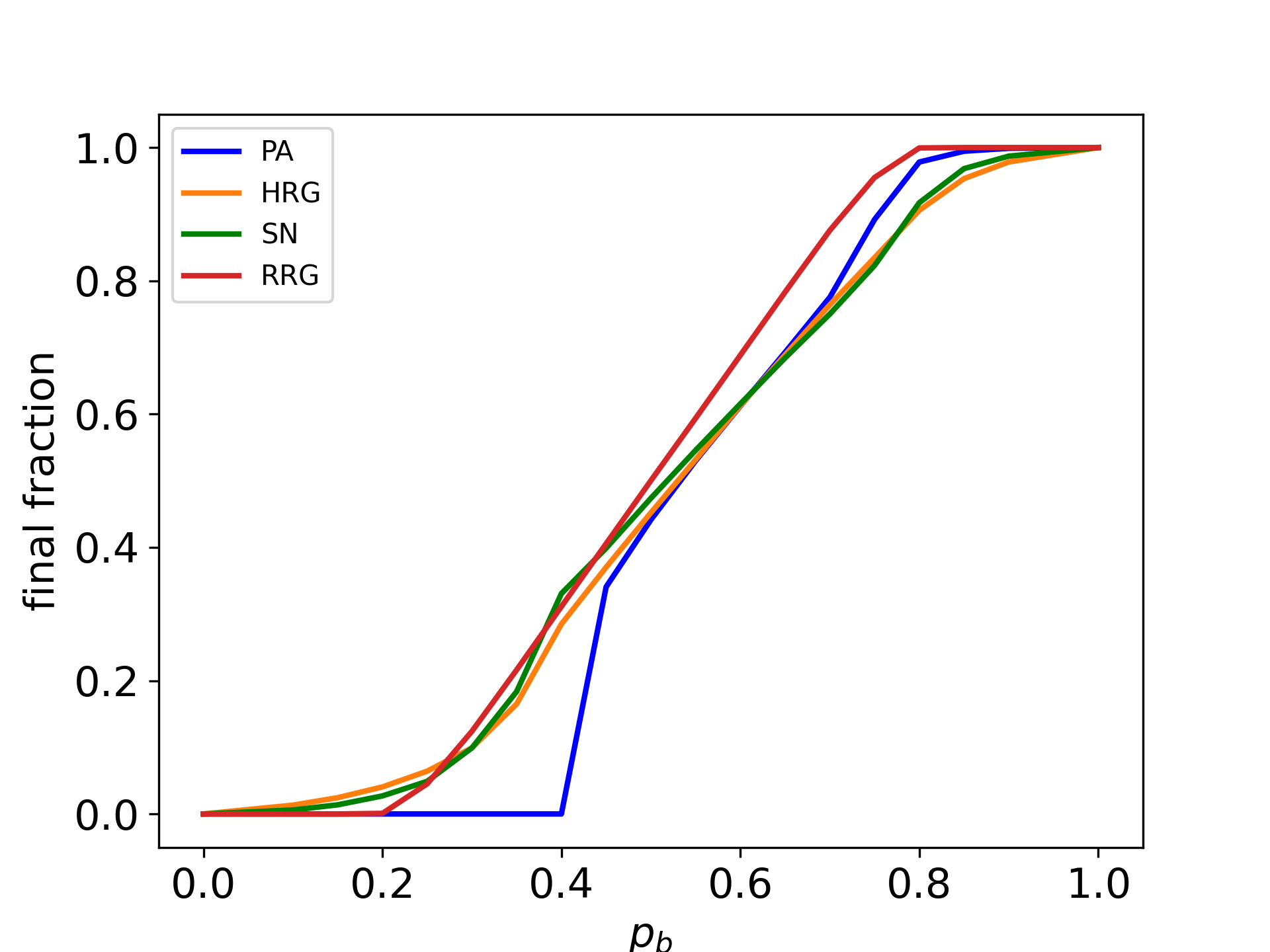}
		\caption{YT SN}
	\end{subfigure}
	\caption{\textbf{(a)} The final fraction of black nodes in the $(\psi_1, \psi_2)$-majority model for $\psi_1 = 0.7$ and $\psi_2 = 0.8$ with a random initial coloring on FB SN and ER, PA, HRG, and RRG with comparable parameters. \textbf{(b)} and \textbf{(c)} provide the same plots for SD SN and YT SN, respectively.}
\end{figure}

\clearpage

\subsection{Expected Stabilization Time Majority Model}
\label{Sec: Experiment 4}

\begin{figure}[H]
	\centering
	\begin{subfigure}{.3\textwidth}
		\includegraphics[scale=0.4]{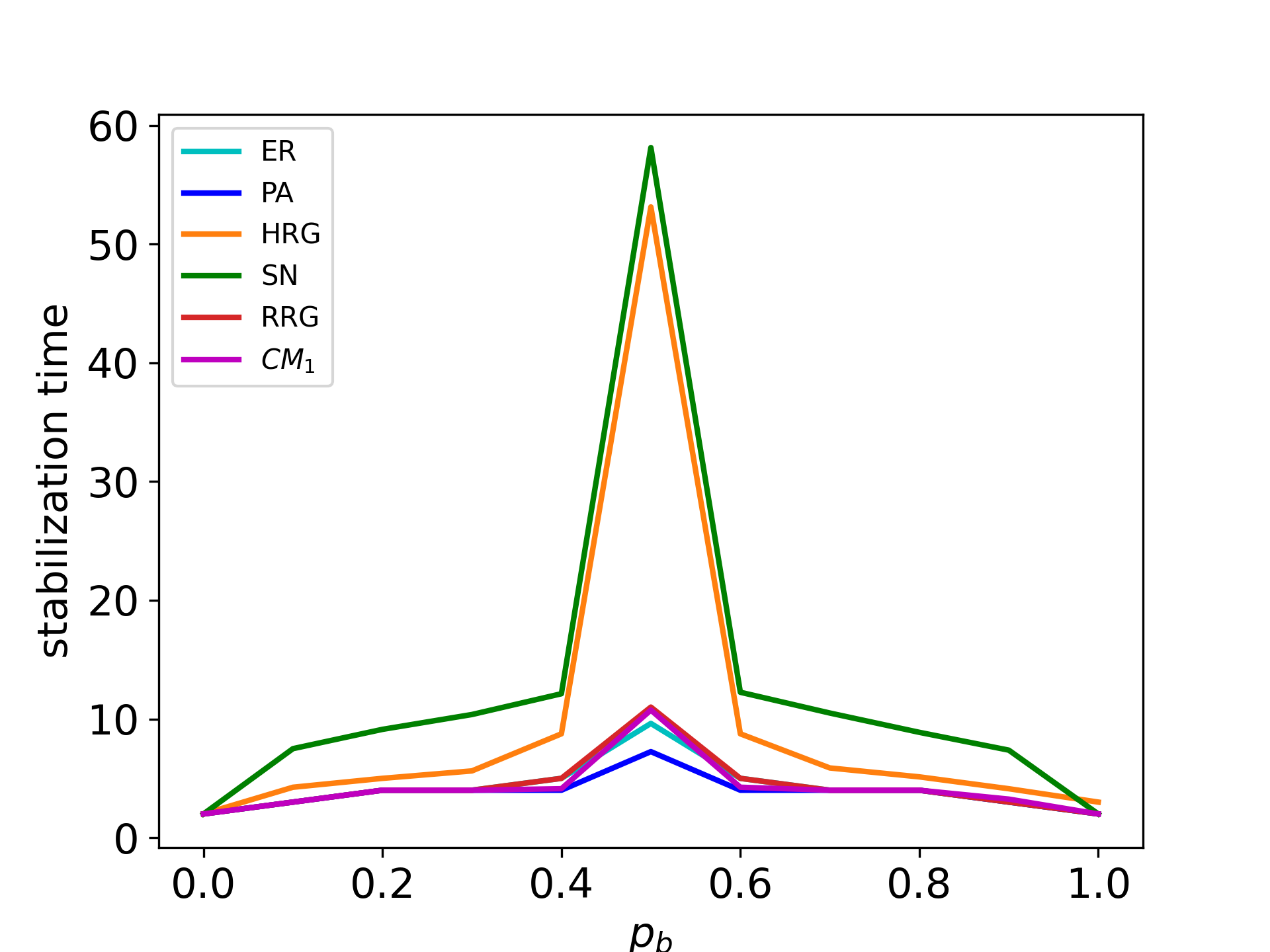}
		\caption{FB SN}
	\end{subfigure}
	\begin{subfigure}{.3\textwidth}
		\includegraphics[scale=0.4]{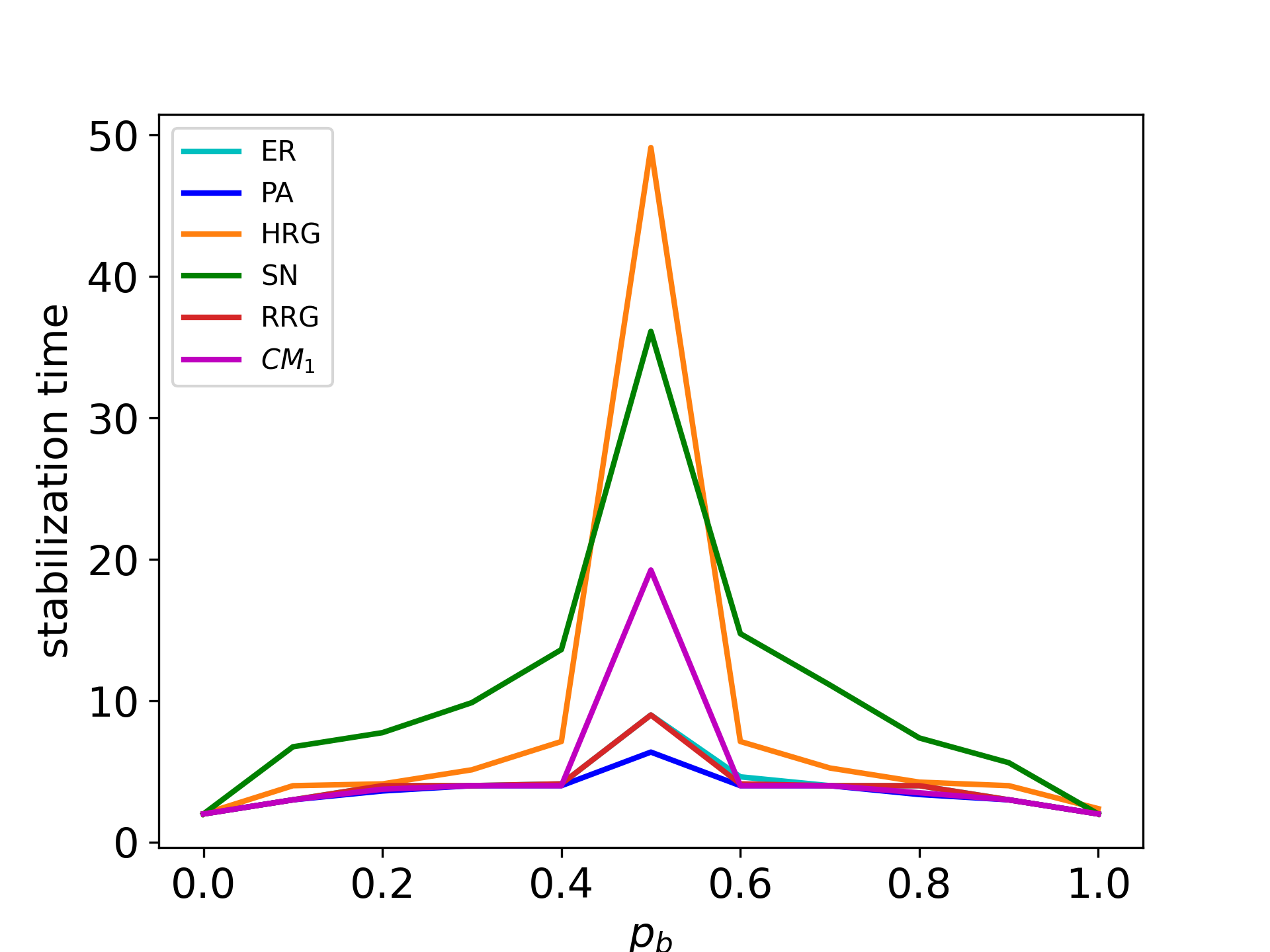}
		\caption{TW SN}
	\end{subfigure}
	\begin{subfigure}{.3\textwidth}
		\includegraphics[scale=0.4]{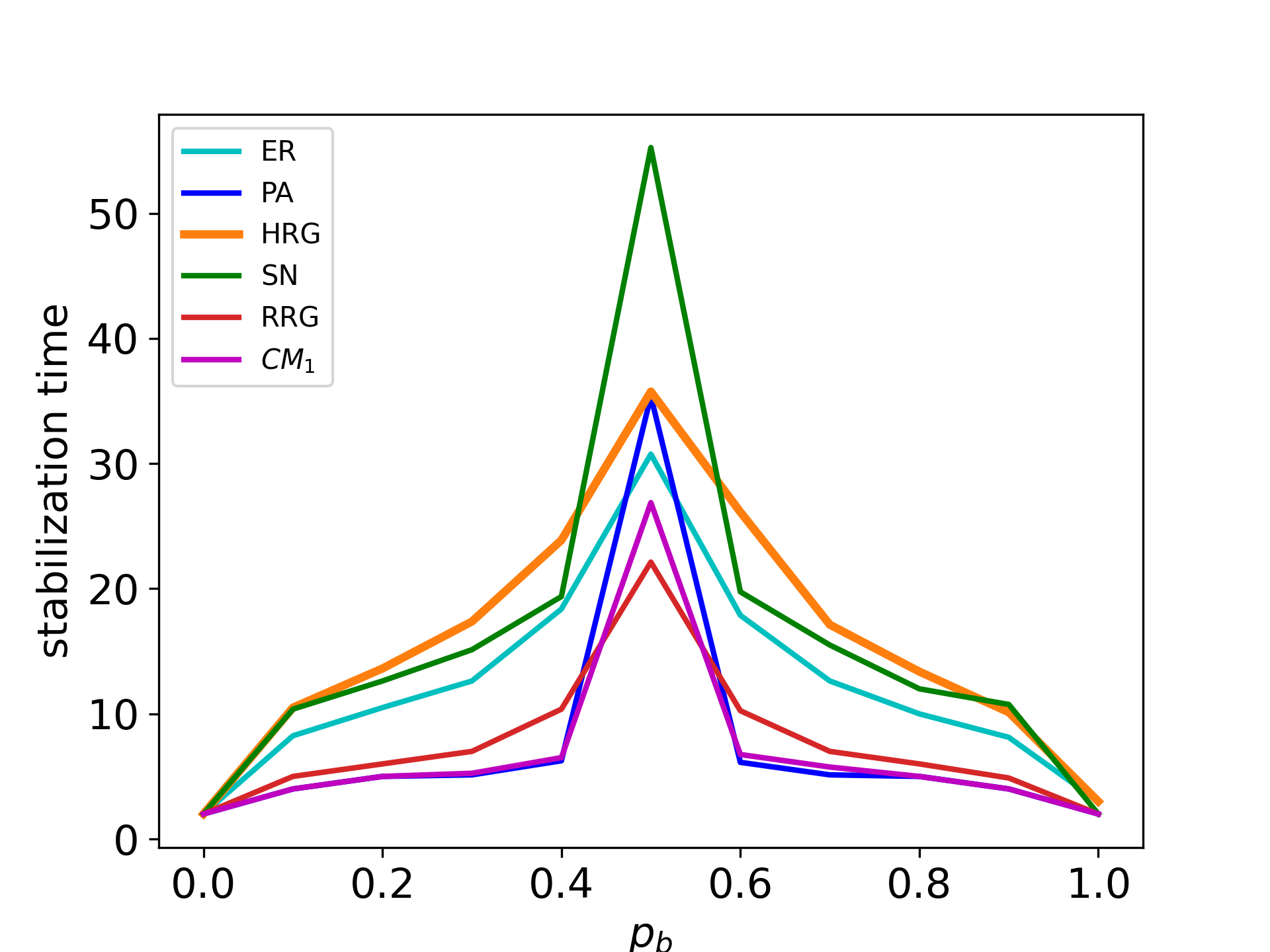}
		\caption{YT SN}
	\end{subfigure}
	\caption{ \textbf{(a)} The expected stabilization time of the majority model with a random initial coloring on FB SN and other graph types with comparable parameters for different values of $p_b$. Here $CM_1$ again denotes the union of FB SN and RRG with degree $d = \bar{d}$ where $\bar{d}$ is the average degree of FB SN. \textbf{(b)} and \textbf{(c)} provide the same plots for TW SN and YT SN, respectively.}
\end{figure}

\vspace{8mm}

\subsection{Expected Stabilization time in the $(\psi_1, \psi_2)$-Majority Model}
\label{Sec: Experiment 5}

\begin{figure}[H]
	\centering
	\begin{subfigure}{.3\textwidth}
		\includegraphics[scale=0.4]{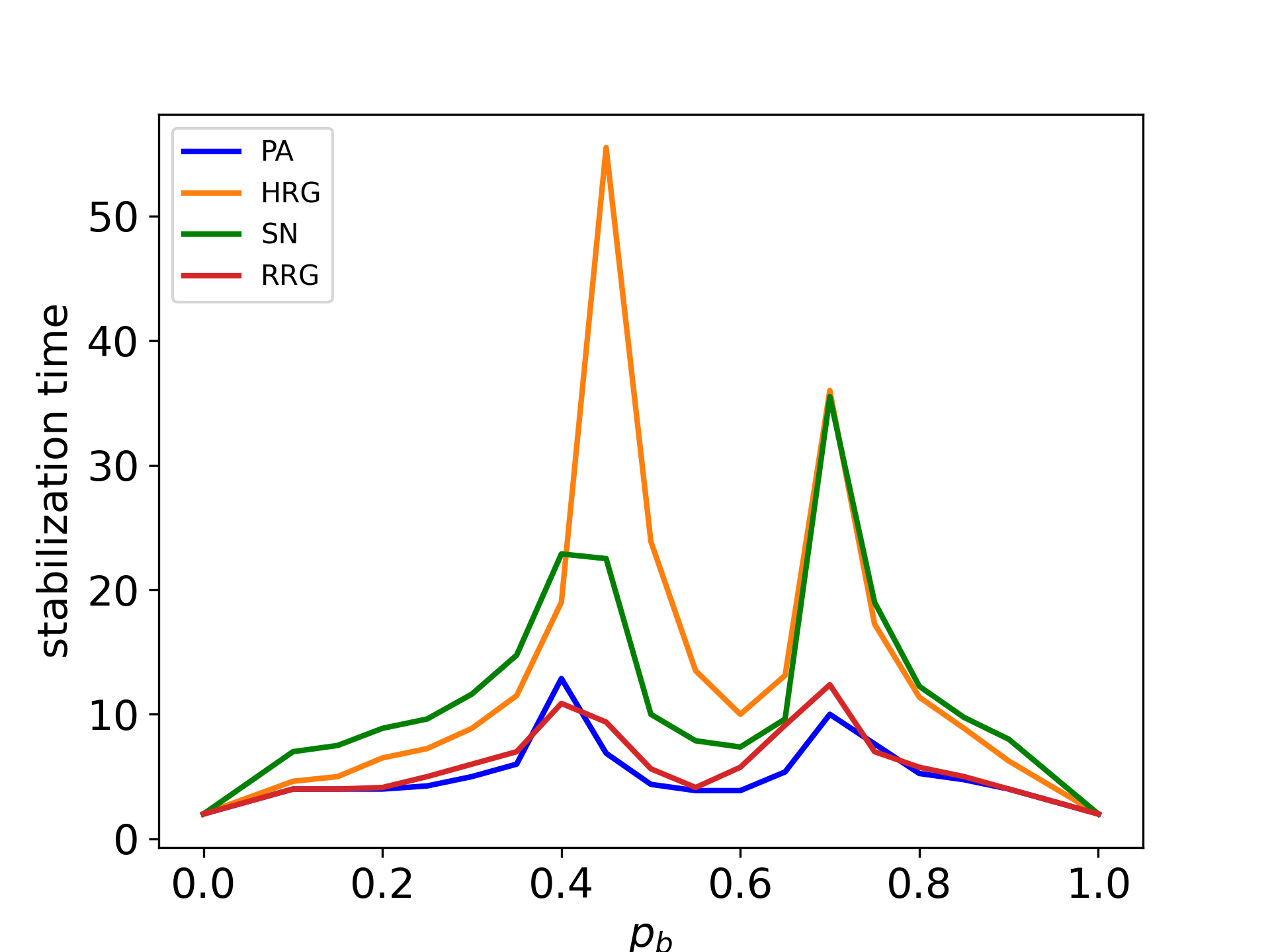}
		\caption{FB SN}
	\end{subfigure}
	\begin{subfigure}{.3\textwidth}
		\includegraphics[scale=0.4]{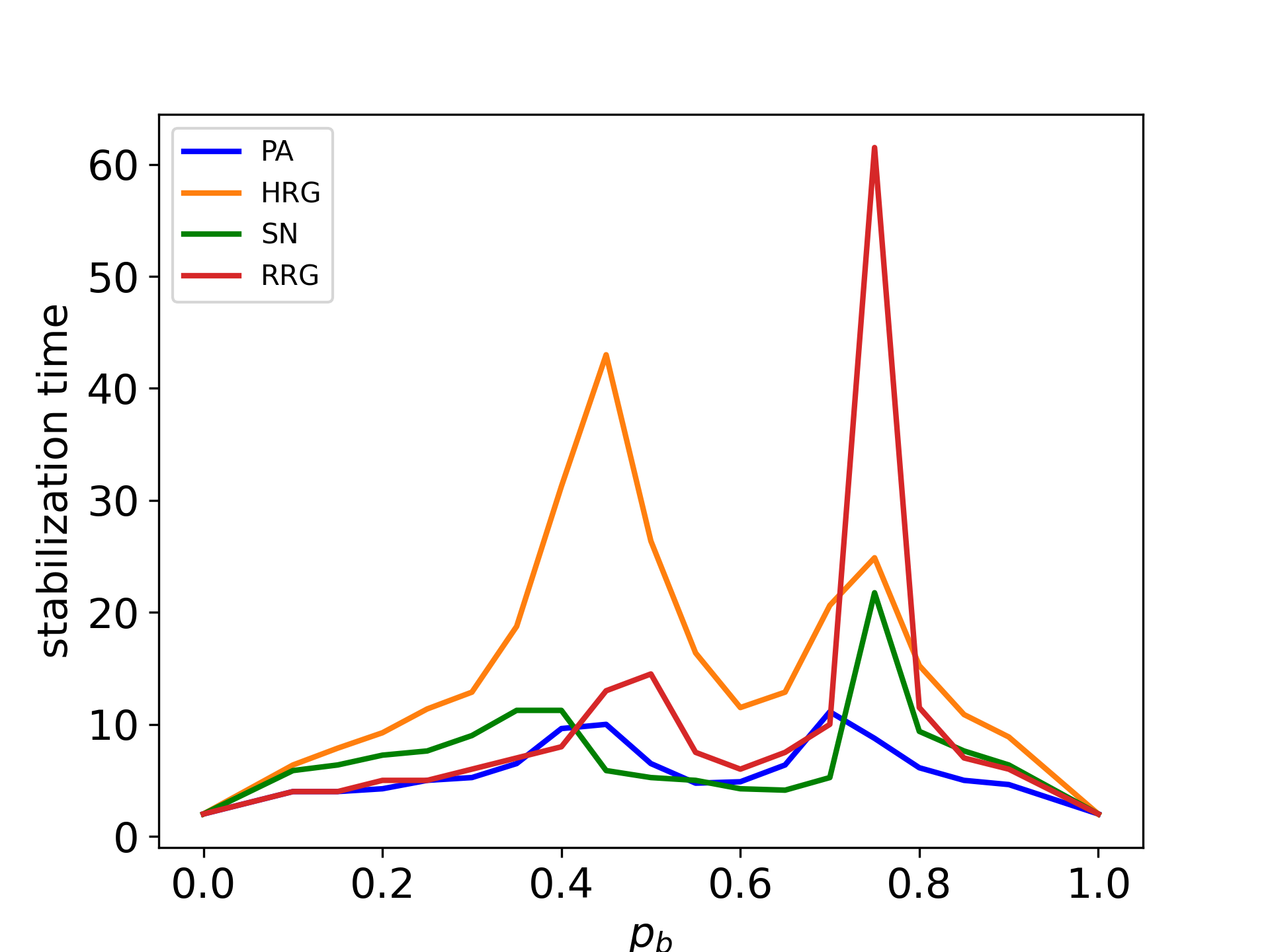}
		\caption{SD SN}
	\end{subfigure}
	\begin{subfigure}{.3\textwidth}
		\includegraphics[scale=0.4]{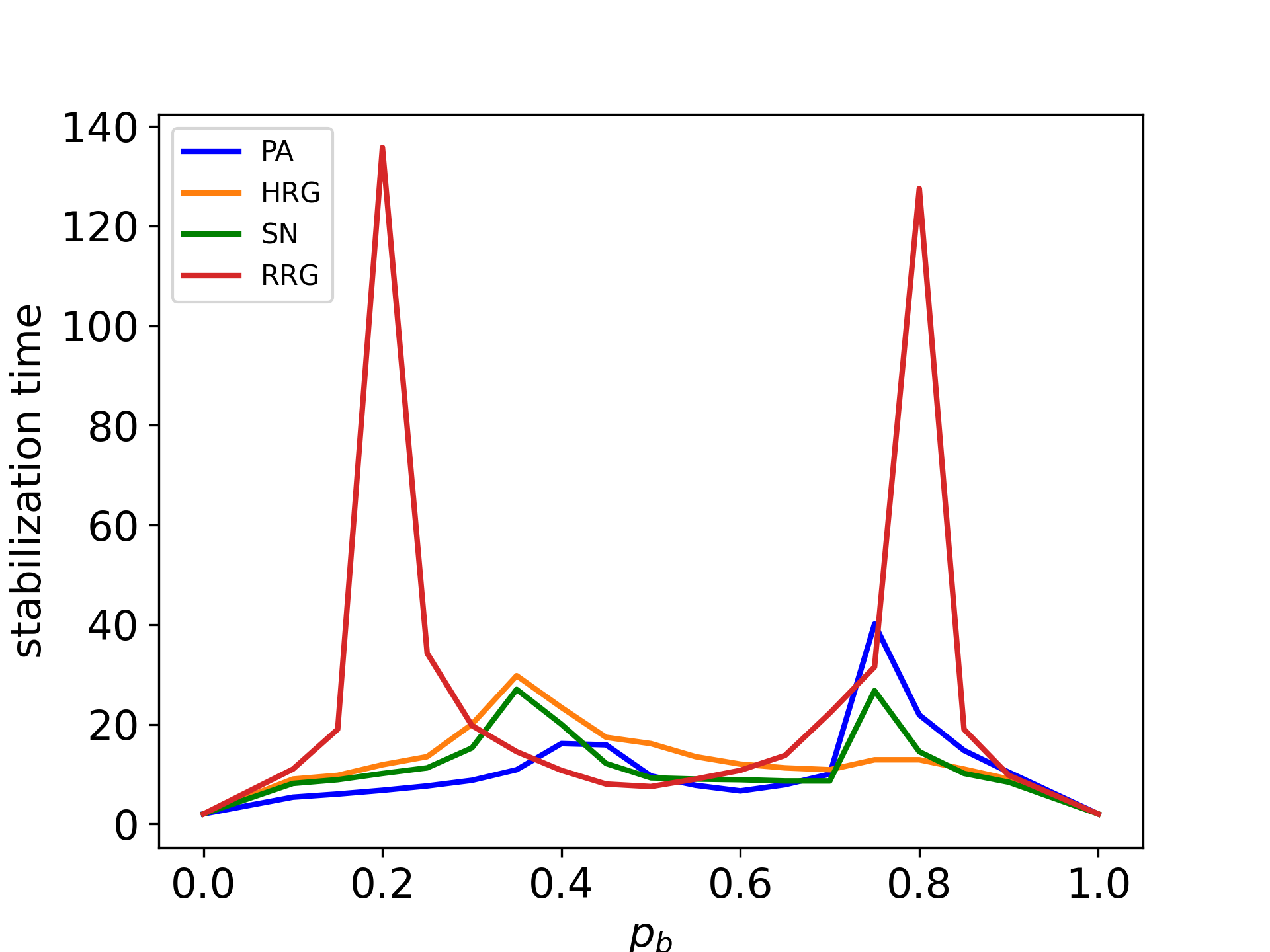}
		\caption{YT SN}
	\end{subfigure}
	\caption{\textbf{(a)} The expected stabilization time of the $(\psi_1, \psi_2)$-majority model for $\psi_1 = 0.7$ and $\psi_2 = 0.8$ on FB SN and other graph types with comparable parameters for different values of $p_b$. \textbf{(b)} and \textbf{(c)} provide the same plots for SD SN and YT SN, respectively.}
\end{figure}

\end{document}